\documentclass{amsart}  

\usepackage{amsfonts}
\usepackage{amsthm}
\usepackage{amsmath}
\usepackage{amsfonts}
\usepackage{latexsym}
\usepackage{amssymb}
\usepackage{amscd}
\usepackage[latin1]{inputenc}
\usepackage{verbatim}
\newtheorem{theorem}{Theorem}[section]
\newtheorem{lemma}[theorem]{Lemma}
\newtheorem{proposition}[theorem]{Proposition}
\newtheorem{corollary}[theorem]{Corollary}
\newtheorem{example}[theorem]{Example}

\theoremstyle{definition}
\newtheorem{definition}{Definition}[section]

\theoremstyle{remark}

\newcommand\style{\mathcal }          %%% calligraphic   
\newcommand\hil{\style H}                        %%% Hilbert space H
\newcommand\hilk{\style K}                       %%% Hilbert space K
\newcommand\bh{\style B(\style H)}    %%%% B(H)
\newcommand\bof{\style B}

\newcommand\md{{\style M}_d}

\newcommand\mpee{{\style M}_p}

\newcommand\tr{ \mbox{Tr} } %%%%%%%% non-normalised trace functional

\newcommand\ucp{{\rm UCP}}
\newcommand\state{{\rm S}}

\newcommand\elem{{\rm E}}                                      %%%%%%%%%%% elementary ucp
\newcommand\epos{{\rm E}^+}                                %%%%%%%%%%% elementary positive Kraus

\newcommand\eff{{\rm Eff}}   %%%% effects
\newcommand\aut{{\rm Aut}}      %%%%%%%%% automorphism group   
 
\newcommand\csta{{\style A}}

\newcommand\bor{{\mathcal O}}     %%%% Borel sets of
\newcommand\povm{{\rm POVM}}     %%%% POVM of
\newcommand\cstarconv{ {\rm C}^*{\rm conv}}                    %%%C*-convex hull of 
\newcommand\cstarconvp{ {\rm C}^*_{\rm p}{\rm conv}}      %%%proper C*-convex hull of 
\newcommand\ext{{\rm ext} }                               %%% extreme points of
\newcommand\cstarext{ {\rm C}^*{\rm ext}}        %%%%  C*-extreme points of
\def\mytimeA#1%format #1 (usually \the\time) using a.m. & p.m.
 {%
   {%
      \count20=#1 %
      \count22=\count20
      \divide \count20 by 60
      \count21=\count20
      \multiply\count21 by -60
      \advance\count22 by \count21
      \ifnum\count20<12
           \def\tapm{ a.m.}%
      \else
           \def\tapm{ p.m.}%
      \fi
      \ifnum\count20=0
           \count20=12
      \fi
      \the\count20:%
      \ifnum\count22<10 0\fi
      \the\count22
   }%
 }

\begin{document}

\title{Classical and Nonclassical Randomness in Quantum Measurements}

  \author{Douglas Farenick}
\address{Department of Mathematics and Statistics, University of Regina,
Regina, Saskatchewan S4S 0A2, Canada}
\email{douglas.farenick@uregina.ca}
  
\author{Sarah Plosker}
\address{Department of Mathematics and Statistics, University of Guelph,
Guelph, Ontario N1G 2W1, Canada}
\email{splosker@uoguelph.ca}

 \author{Jerrod Smith}
\address{Department of Mathematics, University of Toronto,
Toronto, Ontario M5S 2E4, Canada}
\email{smith36j@uregina.ca}  

%%%%%%%%%%%%%%%%%%%%%%%%%%%%%%%%%%%%

\maketitle

\begin{abstract}
The space $\povm_\hil(X)$ of positive operator-valued probability measures on the Borel sets 
of a compact (or even locally compact) Hausdorff space $X$ with values in $\bh$, the algebra of
linear operators acting on a $d$-dimensional Hilbert space $\hil$, is studied from the perspectives of classical and non-classical
convexity through a transform $\Gamma$ 
that associates any positive operator-valued measure $\nu$ with a 
certain completely positive linear map $\Gamma(\nu)$ of the homogeneous C$^*$-algebra $C(X)\otimes\bh$
into $\bh$. 
This association is achieved by using an operator-valued integral
in which non-classical random variables (that is, operator-valued functions) are integrated with respect to positive 
operator-valued measures and which
has the feature that the integral of a random quantum effect is itself a quantum effect. A left inverse $\Omega$ for $\Gamma$
yields an integral representation, along the lines of the classical Riesz Representation Theorem for linear functionals
on $C(X)$, of certain (but not all)
unital completely positive linear maps $\phi:C(X)\otimes\bh\rightarrow\bh$. The extremal and C$^*$-extremal
points of $\povm_\hil(X)$ are determined.
\end{abstract}
 
%%%%%%%%%%%%
\section*{Introduction}%%%%%%%%%%%%%%%%%%%%%%%%%%%%%%%
 
The present paper is a mathematical contribution to quantum probability theory in the setting of finite factors of type I$_{d}$,
the results of which can be understood from the perspectives of non-relativistic quantum mechanics and quantum information theory.

A measurement of a quantum system is represented, mathematically,
by a positive operator-valued probability measure (POVM) $\nu$
defined on a $\sigma$-algebra $\bor(X)$ of measurement events such that whenever a measurement
is made with the system in state $\rho$, the measurement event $E\in\bor(X)$ will occur with probability $\tr(\rho\nu(E))$ \cite{Busch--Lahti--Mittelstaedt-book}.
In this formulation, $X$ is a locally compact Hausdorff space of measurement outcomes, $\bor(X)$ is a $\sigma$-algebra of Borel sets of $X$,
$\rho$ is a density operator acting on a separable Hilbert space $\hil$, and $\tr(\cdot)$ is the
canonical trace on the algebra $\bh$ of bounded linear operators acting on $\hil$. In practice, quantum measurements of an actual physical
system are made by way of some apparatus and in such cases the sample space $X$ is typically
assumed to be finite. Consequently,
a great deal of the literature on the mathematical aspects of POVMs deals only with finite sample spaces $X$. 

On the other hand, probability theory and its use in physics does not require the sample spaces to be finite. Moreover, in theory, 
a POVM defined on an arbitrary (perhaps infinite) sample space $X$ corresponds to a physically realisable quantum measurement.
Therefore, one of our primary goals is to approach the theory
of quantum measurement under the assumptions that $X$ be arbitrary.
More precisely, we consider the fairly general situation in which the sample space $X$ is a compact Hausdorff topological space and the 
Hilbert space $\hil$ has finite dimension $d$.
While this level of abstraction is in accordance with the generalities present in the axioms for measurements of quantum systems,
it is also useful in the mathematical analysis and interpretation of measurements of quantum systems with finitely many outcomes. 
It is not difficult to modify our work herein so that the assumption of the compactness of $X$ be weakened to the requirement that
$X$ be locally compact, but the extension of our work from $d$-dimensional Hilbert space to infinite-dimensional Hilbert space is of 
a very different nature, which we do not intend to address herein.

For a fixed sample space $X$, the set of all measuring apparata of a quantum system $\hil$ is denoted by $\povm_\hil(X)$.
This is a convex set in which a (classical) convex combination 
of POVMs corresponds to a random choice of measuring apparatus. However,
$\povm_\hil(X)$ exhibits a stronger, nonclassical convexity property, namely that of C$^*$-convexity, and we herein consider $\povm_\hil(X)$
from both the classical and nonclassical geometric points of view. In C$^*$-convexity, scalar-valued convex coefficients are replaced by
operator-valued convex coefficients. Thus, a
C$^*$-convex combination of POVMs corresponds to a nonclassical (or quantum)
random choice of apparatus. A set of C$^*$-convex
coefficients corresponds to the noise (or Kraus) operators of a unital quantum channel. Thus,
if a measurement $\nu$ is obtained through a proper C$^*$-convex combination of measurements 
$\nu_1,\dots,\nu_n$, then $\nu$ is a coarser measurement than each of the $\nu_j$. Conceptually, a sharp measurement $\nu$
is one in which the only measurements that are coarser than $\nu$ are those that are (unitarily) equivalent to $\nu$. We shall
show in Theorem \ref{cstar ext pt} that this conceptual notion of sharpness coincides with the common notion of sharpness, namely 
that of a classical observable (or, in POVM terminology, a projection-valued measure) \cite{jencova--pulmannova2008}.

A further nonclassical development made herein is analytical. By taking the word ``measure'' in its literal sense,
we show that for every $\nu\in\povm_\hil(X)$ one may define
an operator-valued integral $\int_X \psi\,d\nu$ for all quantum random variables $\psi$. This integral has the feature that 
the integral of a function whose values are quantum effects is again a quantum effect.
Through the identification of the C$^*$-algebra $C(X)\otimes \bh$ with the unital homogeneous C$^*$-algebra 
of all continuous functions $f:X\rightarrow\bh$, we show that the map $f\mapsto\int_Xf\,d\nu$  
is a unital completely positive (ucp) linear map $\phi_\nu:C(X)\otimes\bh\rightarrow\bh$. 

The representation of $\nu\in\povm_\hil(X)$ by a ucp map $\phi_\nu$ is a function, which we call the $\Gamma$-transform, 
by which the space $\povm_\hil(X)$ can be studied using the theory of completely positive linear maps \cite{Paulsen-book}. 
Because of  the noncommutativity of operator algebra, $\Gamma$ is far from being surjective and does not preserve 
the affine structure of $\povm_\hil(X)$;  
but $\Gamma$ does have a left inverse $\Omega$ which is C$^*$-affine. We endow $\povm_\hil(X)$ with a 
natural topology in which $\povm_\hil(X)$ is a 
compact space. Therefore, all quantum measurements   
are approximated by convex combinations of extremal quantum measurements.

We analyse the extremal POVMs, using some earlier work in \cite{dariano--etal2005,parthasarathy1999} for the case of finite sample spaces,
to obtain a nonclassical version of the classical theorem in probability that asserts that a probability measure is extremal if and 
only if its mass is concentrated at a point of the sample space.  This result (Theorem \ref{ext pt}) can be found in the relatively recent literature
\cite{chiribella--etal2007,chiribella--etal2010}; however, the  approach we take here is rather different and adheres to our
over-arching theme of using ``quantum methods'' for proving statements about quantum probabilities.
Using nonclassical convexity \cite{farenick--morenz1997,farenick--zhou1998}, 
we prove that every quantum measurement
with finitely many outcomes is a coarsening of sharp measurements, and that every quantum measurement with arbitrary outcomes is
statistically approximated by a coarsening of sharp measurements with finitely many outcomes.

%%%%%%%%%%
\section{Notation, Terminology, and Assumptions}

General references for POVMs and completely positive linear maps 
are \cite{Busch--Lahti--Mittelstaedt-book,Davies-book,Holevo-book2} and \cite{Paulsen-book}, respectively.

%%%%%%
\subsection{Assumption}

By $X$ we denote a compact Hausdorff space and by $\hil$ a $d$-dimensional Hilbert space. 
Let $\{e_1,\dots,e_d\}$ be an orthonormal basis for $\hil$, which henceforth is assumed to be fixed.  
For each $i,j\in\{1,\dots,d\}$, we denote by
$e_{ij}\in\bh$ the unique operator that sends $e_j$ to $e_i$ and all other $e_k$ ($k\neq j$) to $0$.

\subsection{States, effects, automorphisms, and POVMs}

The real vector space of all selfadjoint operators acting on $\hil$ is denoted by $\bh_{\rm sa}$ and $\bh_+\subset\bh_{\rm sa}$ denotes 
the cone of positive operators. The state space of $\hil$
is denoted by $\state(\hil)$ and consists of all $\rho\in \bh_+$ of trace $\tr\,\rho=1$. 
In particular, $e\in\state(\hil)$
denotes $\frac{1}{d}1$, where $1\in\bof(\hil)$ is the identity operator.
A quantum effect is a positive operator $h\in\bh_+$ with the property that
$0\leq \lambda\leq 1$
for every eigenvalue $\lambda$ of $h$. The set of quantum effects is denoted by $\eff(\hil)$.
By $\aut(\bh)$ we denote the automorphism group for the C$^*$-algebra $\bh$. Thus,
$\alpha\in\aut(\bh)$ if and only if there is a unitary operator $u\in \bh$ such that
$\alpha(x)=u^*xu$ for every $x\in\bh$.
The $\sigma$-algebra of
Borel sets of $X$ is denoted by $\bor(X)$. 

\begin{definition}\label{povm}
A function $\nu:\bor(X)\rightarrow\bh$ is a \emph{positive operator valued probability measure} (POVM) on $X$
if:
\begin{enumerate}
\item $\nu(E)\in\eff(\hil)$, for every $E\in\bor(X)$;
\item for every countable collection $\{E_k\}_{k\in\mathbb N}\subset\bor(X)$ with $E_k\cap E_{k'}=\emptyset$ for $k'\neq k$, 
\[
\nu\left( \bigcup_{k\in\mathbb N} E_k\right)\,=\,\sum_{k\in\mathbb N}\nu(E_k)\,,
\]
where the convergence of the right hand side of the equality above is with respect to the weak operator topology of $\bh$;
\item $\nu(X)=1\in\bh$.
\end{enumerate}
If, in addition, $\nu(E)$ is a projection for every $E\in\bor(X)$, then $\nu$ is said to be \emph{sharp} (or \emph{classical}).
\end{definition}

In light of the correspondence between physical quantum measurements and mathematical POVMs, we shall frequent use the
terms \emph{quantum measurement} or \emph{quantum instrument} for elements of $\povm_\hil(X)$.

The set of all positive operator valued measures on $X$ with values in $\bh$ is denoted by $\povm_\hil(X)$.
We shall identify the space $P(X)$ of Borel probability measures with the subset 
$\{\mu\cdot1 \,:\,\mu\in P(X)\}\subset\povm_\hil(X)$ so that 
we think of ordinary probability measures as scalar-valued POVMs.

\begin{definition}
The \emph{support} of $\nu\in\povm_\hil(X)$ is the smallest closed subset $K_\nu\subset X$ for which $\nu(X\setminus K_\nu)=0$.
\end{definition}

If the support of $\nu\in\povm_\hil(X)$ consists of a single point, say $K_\nu=\{x_0\}$, then $\nu$ is a \emph{Dirac measure}
and is necessarily of the form $\nu=\delta_{x_0}1$, where $\delta_{x_0}\in P(X)$ satisfies,
for $E\in\bor(X)$, $\delta_{x_0}(E)=1$ if $x_0\in E$ and $\delta_{x_0}(E)=0$ otherwise.

If $\nu\in\povm_\hil(X)$ has finite support $K_\nu=\{x_1,\dots,x_m\}$, then each $h_j=\nu(\{x_j\})\neq0$ and
\[
\nu\,=\,\sum_{j=1}^m\delta_{x_j}h_j\,.
\]

%%%%
\subsection{Completely positive linear maps of homogeneous C$^*$-algebras}
 
If $X$ is compact, then the set $C(X)$ of all continuous functions $X\rightarrow\mathbb C$ is a
unital, abelian C$^*$-algebra, and the C$^*$-algebra $C(X)\otimes\bh$ is naturally identified with the homogeneous
C$^*$-algebra of all continuous functions $f:X\rightarrow\bh$. 

A linear map $\phi:C(X)\otimes\bh\rightarrow\bh$ is completely positive if
\[
\phi\otimes{\rm id}_{\mpee}: C(X)\otimes\bh \otimes\mpee\rightarrow\bh\otimes\mpee
\]
preserves positivity for every algebra $\mpee$ of complex $p\times p$  matrices. A
completely positive linear map that preserves the identity is called a ucp map (unital, completely positive).
Let
\[
\ucp_\hil(X)\,=\,\left\{ \phi\,:\, \phi \;\mbox{ is a ucp map }
C(X)\otimes\bh\rightarrow\bh\right\}\,.
\]
 
\begin{definition} A ucp map  
$\varrho_{x_0}\in\ucp_\hil(X)$ of the form
\[
\varrho_{x_0}(f)\,=\,f(x_0),\;f\in C(X)\otimes\bh,
\]
for some fixed $x_0\in X$, is said to be \emph{spectral}.
\end{definition}

The spectral ucp maps coincide with the ucp maps $\varrho\in\ucp_\hil(X)$
that have the property $\varrho(fg)=\varrho(f)\varrho(g)$ for all $f,g\in C(X)\otimes\bh$.
Any two spectral ucp maps $\varrho_{x_0}, \varrho_{x_1}\in\ucp_\hil(X)$
are unitarily equivalent---that is, $\varrho_{x_1}=\alpha\circ\varrho_{x_0}$
for some automorphism $\alpha\in\aut(\bh)$---if and only if $x_1=x_0$.
Thus, the set
\[
{\rm Sp}_\hil(X)\,=\,\{\varrho_{x_0}\,:\,x_0\in X\}
\]
can be identified with the spectrum \cite{Dixmier-book} of the C$^*$-algebra $C(X)\otimes\bh$.

\begin{definition} A ucp map $\phi\in\ucp_\hil(X)$ of the form
\[
\phi(f)\,=\,\sum_{j=1}^m t_j^*f(x_j)t_j\,,\;f\in C(X)\otimes\bh\,,
\]
for some $x_1,\dots,x_m\in X$ (not necessarily distinct) and 
$t_1,\dots,t_m\in\bh$ is said to be \emph{elementary}.
\end{definition}

The set of all elementary maps $C(X)\otimes\bh\rightarrow\bh$ is denoted by 
$\elem_\hil(X)$, and the set of all elementary maps $f\mapsto\sum_{j=1}^m t_j^*f(x_j)t_j$
in which each $t_j$ is a positive operator is denoted by
$\epos_\hil(X)$. 

Thus, we have a heirarchy:
\[
{\rm Sp}_\hil(X)\, \subset\, \epos_\hil(X) \,\subset\, \elem_\hil(X) \,\subset\,\ucp_\hil(X)\,.
\]

%%%
\subsection{Convexity and C$^*$-Convexity}

The sets $\povm_\hil(X)$ and $\ucp_\hil(X)$ are not only convex, but they are
also C$^*$-convex.
To explain the nonclassical notion of C$^*$-convexity, we consider an abstract context that
will capture the nonclassical convexity of both $\povm_\hil(X)$ and $\ucp_\hil(X)$.

Let $Y$ be a nonempty set
and assume that $V_\hil(Y)$ is the vector space (under pointwise addition and scalar multiplication) 
of all functions $\zeta:Y\rightarrow\bh$. If $z_1,z_2\in\bh$ and $\zeta\in V_\hil(Y)$, then 
define
\[
z_1\zeta z_2:Y\rightarrow\bh\quad\mbox{ given by }\quad y\mapsto z_1\zeta(y)z_2\in\bh\,.
\]
This left and right multiplication of elements of $V$ by elements of $\bh$ gives $V$ the algebraic structure of
a $\bh$-bimodule.

In particular, one may apply unitary similarity transformations to functions $\zeta:Y\rightarrow\bh$.

\begin{definition} Two functions
$\zeta,\zeta'\in V_\hil(Y)$ are \emph{unitarily equivalent} if $\zeta' =u^*\zeta u$ for some unitary operator $u\in\bh$.
\end{definition}

\begin{definition} Assume that $a_1,\dots,a_m\in\bh$ and $\zeta_1,\dots,\zeta_m\in V_\hil(Y)$. Then:
\begin{enumerate}
\item $a_1,\dots,a_m\in\bh$ are called \emph{C$^*$-convex coefficients} if
$\sum_{j=1}^ma_j^*a_j\,=\,1$;
\item a \emph{C$^*$-convex combination} of $\zeta_1,\dots,\zeta_m$ is a function $\zeta:Y\rightarrow\bh$ of the form 
$\zeta=\sum_{j=1}^m a_j^*\zeta_j a_j$,
where $a_1,\dots,a_m\in\bh$ are C$^*$-convex coefficients;
\item a \emph{proper C$^*$-convex combination} of $\zeta_1,\dots,\zeta_m$ is a function $\zeta:Y\rightarrow\bh$ of the form 
$\zeta=\sum_{j=1}^m a_j^*\zeta_j a_j$,
where $a_1,\dots,a_m\in\bh$ are invertible C$^*$-convex coefficients.
\end{enumerate}
Furthermore, a subset $K\subset V_\hil(Y)$ is \emph{C$^*$-convex} over $\bh$ if 
$K$ contains all C$^*$-convex combinations of its elements. 
\end{definition}

The C$^*$-convex sets of interest here are:
\begin{enumerate}
\item $\povm_\hil(X)$, using $Y=\bor(X)$ in the definitions above;
\item $\ucp_\hil(X)$, using $Y=X$; and
\item $\elem_\hil(X)$, a C$^*$-convex subset of $\ucp_\hil(X)$.
\end{enumerate}

\begin{definition} If $R\subset V_\hil(Y)$ is a nonempty subset, 
then
\begin{enumerate}
\item the \emph{C$^*$-convex hull of $R$} is the set $\cstarconv\,R$ consisting of all elements of $V_\hil(Y)$ 
attained from all possible C$^*$-convex combinations of elements of $R$, and
\item the  \emph{proper C$^*$-convex hull of $R$} is the set $\cstarconvp\,R$ consisting of all elements of $V_\hil(Y)$ 
attained from all possible proper C$^*$-convex combinations of elements of $R$.
\end{enumerate}
\end{definition}

Observe that $\cstarconv\,R$ is itself a C$^*$-convex set and that any C$^*$-convex set is also a convex set (in the classical sense).
With this formalism, we have
\[
{\rm Sp}_\hil(X)\, \subset\, \epos_\hil(X) \,\subset\, \elem_\hil(X) \,=\,\mbox{C$^*$-conv}\left( {\rm Sp}_\hil(X)\right)\,\subset\,\ucp_\hil(X)\,.
\]
 
If $K$ is a C$^*$-convex subset of $V_\hil(Y)$ and if $\zeta\in K$, then $u^*\zeta u\in K$ for every unitary operator $u\in\bh$. That is,
$K$ is closed under unitary similarities.
Furthermore, it is easy to show that
every $\zeta\in K$ is a proper C$^*$-convex combination of elements unitarily equivalent to it.
The C$^*$-extremal elements $\zeta$ of $K$ are the ones in which 
this is the only way to represent $\zeta$ as a proper C$^*$-convex combination of other elements of $K$.

\begin{definition} An element $\zeta$ in 
a C$^*$-convex subset $K\subset V_\hil(Y)$ is a \emph{C$^*$-extreme point} if
the only manner in which to express $\zeta$ as a proper C$^*$-convex combination of 
$\zeta_1,\dots,\zeta_m\in K$ is by way of $\zeta_j$ of the form $\zeta_j=u_j^* \zeta u_j$ for some
unitary operators $u_1,\dots,u_m\in \bh$.
\end{definition}

Let the sets of extreme points and C$^*$-extreme points of a C$^*$-convex set $K\subset V_\hil(Y)$ be denoted, respectively, by
\[
\ext\,K\quad\mbox{and}\quad \cstarext\,K\,.
\]
We have the following relationship between the two sets.

\begin{proposition}\label{cstar ext is ext} If $K\subset V_\hil(Y)$ is C$^*$-convex, then $\cstarext\,K\,\subset\, \ext\,K$.
\end{proposition}

\begin{proof} Assume that $\zeta\in K$ is a C$^*$-extreme point and that $\zeta=\lambda\zeta_1+(1-\lambda)\zeta_2$
for some $\zeta_1,\zeta_2\in K$ and some real number $\lambda\in(0,1)$. Set $a_1=\sqrt{\lambda}\,1$
and $a_2=\sqrt{1-\lambda}\,1$ to obtain the proper C$^*$-convex combination
$\zeta=\sum_{j=1}^2a_j^*\zeta_ja_j$. By hypothesis, $\zeta_j=u_j^*\zeta u_j$ for some unitaries $u_1,u_2\in\bh$. Fix $y_0\in Y$ so that
$\zeta(y_0)=\lambda u_1^*\zeta(y_0)u_1+(1-\lambda)u_2^*\zeta(y_0)u_2$. Because $\hil$ has finite dimension,
the operator algebra $\bh$ is a Hilbert space with respect to the Hilbert-Schmidt norm $\|\cdot\|_2$. Hence, the equation 
$\zeta(y_0)=\lambda u_1^*\zeta(y_0)u_1+(1-\lambda)u_2^*\zeta(y_0)u_2$
represents the vector  $\zeta(y_0)$, which lies
on the sphere of radius $\|\zeta(y_0)\|_2$,
as a convex combination of the vectors
$u_1^*\zeta(y_0)u_1$ and $u_2^*\zeta(y_0)u_2$, which also lie on the same sphere.
Because the sphere of a Hilbert space
contains no nontrivial line segments, the vectors $u_1^*\zeta(y_0)u_1$ and $u_2^*\zeta(y_0)u_2$ must coincide with $\zeta(y_0)$.
As this is true for every $y_0$, we obtain $\zeta_1=\zeta_2=\zeta$, and so $\zeta$ is an extreme point of $K$.
\end{proof}

Returning now to $\povm_\hil(X)$ and $\ucp_\hil(X)$, which are the C$^*$-convex sets of interest to us here, we summarise below the current state of knowledge
regarding the extreme and C$^*$-extreme points of these sets.
\begin{enumerate}
\item For arbitrary $X$, the extreme points of $\ucp_\hil(X)$ are deduced from a quite general theorem of Arveson \cite[Theorem 1.4.2]{arveson1969}.
\item For arbitrary $X$, the C$^*$-extreme points of $\ucp_\hil(X)$ can be deduced from a theorem of Farenick and 
      Zhou \cite[Theorem 2.1]{farenick--zhou1998}. 
\item For finite $X$, the extreme points of $\povm_\hil(X)$ have been determined by Parthasarathy \cite{parthasarathy1999}
          and by D'Ariano,  Lo Presti,  and Perinotti \cite{dariano--etal2005}, and for arbitrary $X$ the extreme points of $\povm_\hil(X)$
          are characterised by Chiribella, D'Ariano, and Schlingemann in \cite{chiribella--etal2007,chiribella--etal2010}.
\item For arbitrary $X$, the extreme points  and C$^*$-extreme points  of $\povm_\hil(X)$ are determined in
Theorems \ref{ext pt} and \ref{cstar ext pt} of the present paper.

\end{enumerate}

%%%%%%%%%%%%%%%%%%%%%%%%%%%%%%%%%%%%%%%%%%%%%%%%%%%%%%%%%%%%%%%%%%%%%%%%%%
\section{Quantum Random Variables and Integration}

\begin{definition} A \emph{quantum random variable} is a
 function $\psi:X\rightarrow\bh$ that is \emph{Borel measurable} in the sense that the complex-valued functions
\[
x\mapsto\tr\left(\rho \psi(x)\right)
\]
are Borel measurable for every state $\rho\in\state(\hil)$.
\end{definition}

Equivalently, $\psi:X\rightarrow\bh$ is Borel measurable if,
for every pair of vectors $\xi,\eta\in\hil$, the complex-valued
function $x\mapsto\langle \psi(x)\xi,\eta\rangle$ is Borel measurable. 
Our aim in this section is to define, using the procedure set out in \cite{farenick--zhou2007},
a positive-preserving operator-valued integral $\int_X\psi\,d\nu$ for
any Borel measurable function $\psi:X\rightarrow\bh$ and any
$\nu\in\povm_\hil(X)$.  

Every positive operator $h\in \bh$ has a unique positive square root $h^{1/2}$. Thus, if 
$\psi:X\rightarrow\bh$ is a function for which $\psi(x)$ is
a positive operator for every $x\in X$, then $\psi^{1/2}:X\rightarrow\bh$ denotes the function
$\psi^{1/2}(x)=\left(\psi(x)\right)^{1/2}$.

The following observation will be useful.

\begin{proposition}\label{square root} If $\psi:X\rightarrow\bh$ is a positive quantum random variable, then
$\psi^{1/2}$ is a (positive) quantum random variable.
\end{proposition}

\begin{proof} Assume first that $\psi(x)$ is positive and invertible 
for every $x\in X$. Because sums and products of scalar-valued measurable functions 
are measurable, if one invokes an iterative procedure to compute $\psi(x)^{1/2}$---such as the one in
\cite[Algorithm 2]{higham1997}, which is a Newton-type iteration combined with a
Cholesky factorisation---then for each state $\rho\in\state(\hil)$ the function
\[
x\,\longmapsto\,\tr\left(\rho\psi(x)^{1/2}\right)
\]
is a pointwise limit of a sequence of measurable functions. Thus, $\psi^{1/2}$ is a quantum random variable.
In the case where $\psi(x)$ is not invertible for all $x\in X$, 
then $\psi^{1/2}$ is a pointwise limit of $x\mapsto (\psi(x)+\frac{1}{n}1)^{1/2}$
and, hence, is measurable.
\end{proof}
 
%%%%%%%%%%%%%%%%
\subsection{The Principal Radon-Nikod\'ym Derivative}

\begin{definition} If $\nu_1,\nu_2\in\povm_\hil(X)$, then
$\nu_2$ is \emph{absolutely continuous} with respect to $\nu_1$, denoted by
$\nu_2 \ll_{\rm ac} \nu_1$, if $\nu_2(E)=0$ for every $E\in\bor(X)$ for which $\nu_1(E)=0$.
\end{definition} 

 If $\nu\in\povm_\hil(X)$, then a probability measure $\mu$ is obtained from $\nu$ via
\begin{equation}\label{induced}
\mu(E)\,=\,\frac{\tr\left(\nu(E)\right)}{d},\;\mbox{for every }E\in\bor(X)\,.
\end{equation}
Because the trace functional maps nonzero positive operators to strictly positive real numbers, the measures $\mu$
and $\nu$ are mutually absolutely continuous:  $\mu \ll_{\rm ac} \nu$ and $\nu \ll_{\rm ac} \mu$. 

Recall that $\{e_1,\dots,e_d\}$ is a fixed orthonormal basis of $\hil$. Because $\nu \ll_{\rm ac} \mu$,
each of the $d^2$ complex measures $\nu_{ij}:\bor(X)\rightarrow\mathbb C$, defined by
$\nu_{ij}(E)=\langle\nu(E)e_j,e_i\rangle$, has the property that $\nu_{ij}\ll_{\rm ac}\mu$. 
Hence, by the (classical) Radon-Nikod\'ym Theorem, there is a unique 
$\displaystyle\frac{d\nu_{ij}}{d\mu}\in L^1(X,\mu)$
such that
\[
\nu_{ij}(E)\,=\,\int_E \displaystyle\frac{d\nu_{ij}}{d\mu}\,d\mu,\;\mbox{for all }E\in\bor(X)\,.
\]
These scalar Radon-Niko\'ym derivatives give rise to an operator-valued Borel function 
$\displaystyle\frac{d\nu}{d\mu}: X\rightarrow \bh$ via
\begin{equation}\label{p r-n}
\displaystyle\frac{d\nu}{d\mu}\,=\,\sum_{i,j=1}^d \displaystyle\frac{d\nu_{ij}}{d\mu}\otimes e_{ij}\,.
\end{equation}
Notice that for any $\xi=\displaystyle\sum_{k=1}^d\xi_k e_k\in\hil$,
\[
\left\langle \displaystyle\frac{d\nu}{d\mu}\xi,\xi \right\rangle \,=\, \displaystyle\sum_{i,j=1}^d{ \frac{d\nu_{ij}}{d\mu}\xi_j\overline{\xi_i}}\,.
\]
Hence, for all $\xi\in\hil$ and $E\in\bor(X)$, 
\[
\displaystyle\int_E\left\langle\displaystyle\frac{d\nu}{d\mu}(x)\xi,\xi \right\rangle\,d\mu(x) \,=\,
\displaystyle\sum_{i,j=1}^d\left( \displaystyle\int_E \frac{d\nu_{ij}}{d\mu}\,d\mu\right)\xi_j\overline{\xi_i}
\,=\,\langle \nu(E)\xi,\xi \rangle \,\geq\,0\,.
\]
This proves that $\displaystyle\frac{d\nu}{d\mu}(x)$ is a positive operator for
$\mu$-almost all $x\in X$; for such $x$ let 
$\left(\displaystyle\frac{d\nu}{d\mu}(x)\right)^{1/2}$ denote the positive square root (in $\bh$) of the positive
operator $\displaystyle\frac{d\nu}{d\mu}(x)$. Now define 
$\left(\displaystyle\frac{d\nu}{d\mu}\right)^{1/2}:X\rightarrow\bh$
to be $\left(\displaystyle\frac{d\nu}{d\mu}(x)\right)^{1/2}$  at those
$x\in X$ for which $\displaystyle\frac{d\nu}{d\mu}(x)$ is a positive operator,
and zero otherwise.

\begin{definition} If  $\nu\in\povm_\hil(X)$ and if $\mu\in P(X)$ is the induced classical probability measure defined in 
\eqref{induced}, then the Borel function $\displaystyle\frac{d\nu}{d\mu}$ defined in \eqref{p r-n} is
called the \emph{principal Radon-Nikod\'ym derivative of $\nu$}.
\end{definition} 

Unlike the classical case, whenever $d>1$ the principal Radon-Nikod\'ym derivative of $\nu$ depends on the pre-selected choice
of orthonormal basis $\{e_1,\dots,e_d\}$ of $\hil$. If one had chosen a different orthonormal basis, say $\{e_1',\dots,e_d'\}$, then the
resulting principal Radon-Nikod\'ym derivative computed in this new basis is simply that of $\alpha\circ\nu$ in the originally selected basis,
where $\alpha$ is the automorphism induced by the unitary operator that transforms the basis 
$\{e_1',\dots,e_d'\}$ to the basis $\{e_1,\dots,e_d\}$.
The following proposition is even more general.

Recall that a \emph{unital quantum channel} is a linear map $\mathcal E:\bh\rightarrow\bh$ such that $\mathcal E$ is unital, completely
positive, and trace preserving (that is, $\tr\circ\mathcal E=\tr$). Note that $\mathcal E\circ\nu\in\povm_\hil(X)$
for every $\nu\in\povm_\hil(X)$.

\begin{proposition}\label{channels} Assume that $\nu\in\povm_\hil(X)$ and that $\mathcal E:\bh\rightarrow\bh$ is a unital quantum channel.
Let $\mu^\nu$ and $\mu^{\mathcal E\circ\nu}$ be the probability measures induced by $\nu$ and $\mathcal E\circ\nu$ in
accordance with \eqref{induced}. Then there is a $\mu\in P(X)$ such that 
\begin{enumerate}
\item $\mu=\mu^\nu=\mu^{\mathcal E\circ\nu}$ and
\item $\displaystyle\frac{d(\mathcal E\circ\nu)}{d\mu}\,=\,\mathcal E\circ\displaystyle\frac{d\nu}{d\mu}$.
\end{enumerate}
\end{proposition}

\begin{proof} The channel $\mathcal E$ is trace preserving, so for any $E\in\bor(X)$ 
\[
\mu^{\mathcal E\circ\nu}(E)=\displaystyle\frac{1}{d}\tr(\mathcal E(\nu(E)))=\displaystyle\frac{1}{d}\tr(\nu(E))=\mu^\nu(E)\,.
\]
The desired measure is $\mu=\mu^{\mathcal E\circ\nu}=\mu^\nu$.

Let $a=\sum_{i,j=1}^d \alpha_{ij}e_{ij}\in\bh$ and consider $a^*\nu a$.  If $\mu(E)=0$,
then $\nu(E)=0$ and $a^*\nu(E) a=0$; therefore $a^*\nu a \ll_{\rm ac} \mu$.  Fix $i,j$ and consider the $(i,j)$-coordinate measure of $a^*\nu a$:
\[ 
\omega_{ij} = \displaystyle\sum_{l=1}^d \displaystyle \sum_{k=1}^d \overline{\alpha_{li}}\alpha_{kj}\nu_{lk}.
\]
Since $a^*\nu a \ll_{\rm ac} \mu$, we have $\omega_{ij}\ll_{\rm ac}\mu$ and so we may consider the Radon-Nikod\'ym derivative
\[
\displaystyle\frac{d\omega_{ij}}{d\mu}=\displaystyle\sum_{l=1}^d \displaystyle \sum_{k=1}^d \overline{\alpha_{li}}\alpha_{kj}\left(\displaystyle\frac{d\nu_{lk}}{d\mu}\right)=\left(a^* \displaystyle\frac{d\nu}{d\mu} a\right)_{ij}.
\]
Therefore, $\displaystyle\frac{d(a^*\nu a)}{d\mu} = a^*\displaystyle\frac{d\nu}{d\mu}a$.

Consider the Kraus decomposition of the channel $\mathcal E$:
\[
\mathcal E(y)=\displaystyle\sum_{j=1}^{q}{a_j^*ya_j}, \; y\in\bh,\;\mbox{where } \displaystyle\sum_{j=1}^q{a_j^*a_j} = \displaystyle\sum_{j=1}^q{a_ja_j^*}=1\,.
\]
By linearity of the scalar Radon-Nikod\'ym derivative,  $\displaystyle\frac{d(\mathcal E\circ\nu)}{d\mu}\,=\,\mathcal E\circ\displaystyle\frac{d\nu}{d\mu}$.
\end{proof}

%%%%%%%%%%
\subsection{Integrable Functions}

If $f,\psi:X\rightarrow\bh$ are quantum random variables such that $\psi(x)\in\bh_+$ for all $x\in X$, then $\psi^{1/2}$
is measurable (Proposition \ref{square root}) and, thus, the function
$\psi^{1/2}f\psi^{1/2}$ is Borel measurable.

\begin{definition} Assume that $\nu\in\povm_\hil(X)$ and that
$\displaystyle\frac{d\nu}{d\mu}$ is the principal Radon-Nikod\'ym 
derivative of $\nu$.
\begin{enumerate}
\item
If $f:X\rightarrow\bh$ is a Borel function, then 
$f$ is said to be \emph{$\nu$-integrable} if, for every state $\rho\in\state(\hil)$,
the complex-valued function 
\[
f_\rho(x)\,=\,\tr\left(\rho\,\left(\displaystyle\frac{d\nu}{d\mu}(x)\right)^{1/2}f(x)\left(\displaystyle\frac{d\nu}{d\mu}(x)\right)^{1/2}\right)\,,\;x\in X,
\]
is $\mu$-integrable.
\item
The \emph{integral} of a $\nu$-integrable function $f:X\rightarrow\bh$
is defined to be the unique operator acting on $\hil$ having the property that
\[
\tr\left(\rho\int_Xf\,d\nu\right)\,=\,\int_X\,f_\rho\,d\mu
\]
for every state $\rho$ of $\hil$.
\end{enumerate}
\end{definition}

\begin{example} The integral of an effect-valued function is an effect.
\end{example}
To verify this claim, choose $\nu\in\povm_\hil(X)$ and let $\mu$ be its principal Radon-Nikod\'ym derivative.
Assume that $f:X\rightarrow\eff(\hil)$ is $\nu$-integrable. Because $0\leq f(x)\leq 1$ in $\bh_{\rm sa}$,
for every state $\rho\in\state(\hil)$ we have
\[
0\,\leq\,
\rho^{1/2}\,\left(\displaystyle\frac{d\nu}{d\mu}(x)\right)^{1/2}f(x)\left(\displaystyle\frac{d\nu}{d\mu}(x)\right)^{1/2}\rho^{1/2}
\,\leq\,
\rho^{1/2}\,\left(\displaystyle\frac{d\nu}{d\mu}(x)\right)\rho^{1/2}
\]
for $\mu$-almost all $x\in X$. Thus, for every $\rho\in\state(\hil)$, we have
\[
\int_X f_\rho\,d\mu \,\leq\, \int_X \tr\left(\rho \frac{d\nu}{d\mu}\right)\,d\mu
\]
and so $0\,\leq\,\displaystyle\int_X f\,d\nu\,\leq\,\int_X \left( \frac{d\nu}{d\mu}\right)\,d\nu\,=\,\nu(X)\,=\,1\in\bh$.
\hfill$\diamondsuit$

\begin{example} The principal Radon-Nikod\'ym derivative of $\nu=\displaystyle\sum_{j=1}^n \delta_{x_j}h_j$
and the corresponding integral formula.
\end{example}
Here, we assume that $h_1,\dots,h_n\in\bh_+$ are nonzero and satisfy $\sum_j h_j=1$ and that $\{x_1,\dots,x_n\}$ is a set of $n$
distinct points of $X$. The measurement $\nu\in\povm_\hil(X)$ is defined by 
\[
\nu(E)\,=\,\sum_{j=1}^n \delta_{x_j}(E)h_j,\; E\in\bor(X)\,.
\]
If $\chi_E$ denotes the characteristic (or indicator) function of any measurement event $E\in\bor(X)$, then
\[
\frac{d\nu}{d\mu}\,=\,\sum_{j=1}^n \left( \frac{d}{\tr(h_j)}\,\chi_{\{x_j\}}\right) h_j
\]
and
\[
\int_X f\,d\nu\,=\, \sum_{j=1}^n h_j^{1/2}f(x_j) h_j^{1/2}\,,
\]
for every Borel function $f:X\rightarrow\bh$. 
\hfill$\diamondsuit$
\vskip 6 pt

%%%%%%%%%%
\subsection{Quantum Integration is a Completely Positive Operator}

The following theorem is the first main result of the present paper.
To set the notation used in the proof, for any operator algebra $\csta$ we let
$\md(\csta)$ denote the C$^*$-algebra of $d\times d$ matrices with entries from $\csta$.
An element $F\in\md(\csta)$ is a matrix $F=[f_{k\ell}]_{k,\ell=1}^d$ of elements $f_{k\ell}\in\csta$.

\begin{theorem}\label{pairing} If 
$\nu\in\povm_\hil(X)$, then there is a unital completely positive linear map 
\[
\phi_\nu:C(X)\otimes\bh\rightarrow\bh
\]
such that
\[
\phi_\nu(f)\,=\,\int_X f\,d\nu\,,
\]
for every $f\in C(X)\otimes\bof(\hil)$. 
\end{theorem}

\begin{proof} 
Choose $\nu\in\povm_\hil(X)$ and let $\mu$ be its principal Radon-Nikod\'ym derivative.
Because $\hil$ has finite dimension $d$, to prove that $\phi_\nu$ is completely positive it is sufficient to show that
the linear function $s_{\phi_\nu}:\md\left(C(X)\otimes\bh\right)\rightarrow\mathbb C$ defined by
\[
s_{\phi_\nu}\left([f_{k\ell}]_{k,\ell}\right)\,=\,
\frac{1}{d}\sum_{k,\ell=1}^d \langle \phi_\nu(f_{k\ell})e_\ell,e_k\rangle
\]
maps positive elements of $\md\left(C(X)\otimes\bh\right)$ to nonnegative real numbers \cite[Theorem 6.1]{Paulsen-book}.
Note that $G=[g_{k\ell}]_{k,\ell}\in \md\left(C(X)\otimes\bh\right)$ is positive if, for every $x\in X$, the operator
\[
G(x)\,=\,[g_{k\ell}(x)]_{k,\ell} \,\in\,\md\left(\bh\right)\,=\,\bof(\hil\otimes\mathbb C^d)\,=\,\bof(\bigoplus_{1}^d\hil)
\]
is positive.

Thus, let $F=[f_{k\ell}]_{k,\ell}\in \md\left(C(X)\otimes\bh\right)$ be positive. The principal Radon-Nikod\'ym derivative is positive $\mu$-almost
everywhere, and so for $\mu$-almost all $x\in X$, the operator matrix $G(x)=K(x)^{1/2}F(x)K(x)^{1/2}$ is a positive operator acting on 
$\bigoplus_1^d\hil$, where $K(x)$ is the diagonal operator matrix with diagonal entries $\frac{d\nu}{d\mu}(x)$. 
Therefore, the $(k,\ell)$-entry of $G(x)$ is $g_{k\ell}(x)=\left(\frac{d\nu}{d\mu}(x)\right)^{1/2}f_{k\ell}(x)\left(\frac{d\nu}{d\mu}(x)\right)^{1/2}$
for $\mu$-almost all $x$. In particular, with $\xi=e_1\oplus\cdots\oplus e_d\in\bigoplus_1^d\hil$, we have
\[
0\,\leq\,\langle G(x)\xi,\xi\rangle\,=\, \sum_{k,\ell=1}^d \langle g_{k\ell}(x)e_\ell,e_k\rangle
\]
for $\mu$-almost all $x\in X$.
Hence, assuming $F=[f_{k\ell}]_{k,\ell}\in \md\left(C(X)\otimes\bh\right)$ is positive, we deduce that
\[
\begin{array}{rcl}
s_{\phi_\nu}\left([f_{k\ell}]_{k,\ell}\right) &=& 
\displaystyle\frac{1}{d}\sum_{k,\ell=1}^d \left\langle \left(\displaystyle\int_Xf_{k\ell}\,d\nu\right)e_\ell,e_k\right\rangle \\ && \\
&=& \displaystyle\frac{1}{d}\displaystyle\int_X\left(\sum_{k,\ell=1}^d \langle g_{k\ell}(x)e_\ell,e_k\rangle\right)\,d\mu(x) \\ && \\
&\geq& 0\,.
\end{array}
\]
That is, $\phi_\nu$ is completely positive. Lastly,
because
\[
\phi_\nu(1)\,=\,\int_X1\,d\nu\,=\,\nu(X)\,=\,1\,,
\]
we conclude that $\phi_\nu$ is a unital map.
\end{proof}

%%%%%%%%%%%%%
\subsection{The $\Gamma$-transform}

We now formalise the association of $\phi_\nu\in\ucp_\hil(X)$ with $\nu\in\povm_\hil(X)$.

\begin{definition} Define $\Gamma:\povm_\hil(X)\rightarrow\ucp_\hil(X)$ by
\[
\Gamma(\nu)\,=\,\phi_\nu\,.
\]
\end{definition}

Because the definition of $\phi_\nu(f)$ involves a square root of the principal Radon-Nikod\'ym derivative $\displaystyle\frac{d\nu}{d\mu}$, 
the transform $\Gamma$ does not appear to possess any usable affine properties. However, we are able to say how $\Gamma(\nu)$ and
$\Gamma(\nu')$ compare if $\nu'$ is obtained from $\nu$ via composition with an automorphism.

\begin{proposition}\label{aut1} For every $\nu\in\povm_\hil(X)$ and $\alpha\in\aut(\bh)$, the following equation holds:
\[
\Gamma(\alpha\circ\nu)\,=\,\alpha\circ\Gamma(\nu)\circ[{\rm id}_{C(X)}\otimes\alpha^{-1}]\,.
\]
That is, for every $\nu\in\povm_\hil(X)$ and unitary $u\in \bh$, we have
\[
\int_X f\,d(u^*\nu u)\,=\,u^*\left(\int_X ufu^*\,d{\nu}\right)u\,,
\]
for every continuous $f:X\rightarrow\bh$.
\end{proposition}

\begin{proof}
 By Proposition \ref{channels}, 
$\displaystyle\frac{d(u^*\nu u)}{d\mu}=u^*\displaystyle\frac{d\nu}{d\mu} u$, 
where $\mu=\displaystyle\frac{1}{d}\tr\circ\nu=\displaystyle\frac{1}{d}\tr\circ\alpha\circ\nu$.
Recall that $\displaystyle\int_X f d(u^*\nu u)$ is the unique operator such that, for any state $\rho\in\state(\hil)$,
\[
\begin{array}{rcl}
\tr\left( \rho \displaystyle\int_X f d(u^*\nu u) \right) &=& 
\displaystyle \int_X \tr \left( \rho \left(\displaystyle\frac{d(u^*\nu u)}{d\mu}\right)^{1/2} f \left (\displaystyle\frac{d(u^*\nu u)}{d\mu}\right)^{1/2}\right) \,d\mu
\\ && \\
&=& \displaystyle \int_X \tr \left( \rho u^* \left(\displaystyle\frac{d\nu}{d\mu}\right)^{1/2} u f u^*  \left (\displaystyle\frac{d\nu }{d\mu}\right)^{1/2} u\right) \, d\mu
\\ && \\
&=& \displaystyle \int_X \tr \left(u \rho u^* \left(\displaystyle\frac{d\nu}{d\mu}\right)^{1/2} u f u^*  \left (\displaystyle\frac{d\nu }{d\mu}\right)^{1/2} \right) \, d\mu
\\  && \\
&=& \tr \left( u\rho u^* \left( \displaystyle \int_X u f u^* d\nu \right) \right)
\\ && \\
&=& \tr \left( \rho u^* \left( \displaystyle \int_X u f u^* d\nu \right) u \right) \,.
\end{array}
\]
Hence,
$\displaystyle
\int_X f\,d(u^*\nu u)\,=\,u^*\left(\int_X ufu^*\,d{\nu}\right)u$.
\end{proof}

If one were to seek a similar reformulation of Proposition \ref{aut1}
by replacing the automorphism 
$\alpha$ with a unital, invertible quantum channel $\mathcal E$, then 
at a purely formal level one would anticipate that
\begin{equation}\label{iqc}
\Gamma(\mathcal E\circ\nu)\,=\,\mathcal E\circ\Gamma(\nu)\circ[{\rm id}_{C(X)}\otimes\mathcal E^{-1}]\,.
\end{equation}
Note, however, that such a formulation should not require $\mathcal E^{-1}$ to be 
completely positive, for if it were a requirement, then $\mathcal E$ would be an automorphism
(\cite[Corollary 2.3.2]{Davies-book}, \cite[Theorem X.5]{clean2005}), which brings us back
to the case of Proposition \ref{aut1}.
Whether equation \eqref{iqc} holds for arbitrary invertible unital channels $\mathcal E$
remains open.

 %%%%%%%
\subsection{Non-Principal Radon-Nikod\'ym Derivatives}

If $h\in\bh_+$, then $h^{-1}$ shall denote the unique
positive operator for which $\ker h^{-1}=\ker h$ and
$h^{-1}h=hh^{-1}=q$, the projection onto the range of $h$. Thus, if $h$ is invertible, then $h^{-1}$ 
is the inverse of $h$. Once we have this notion for positive operators, a similar notion of generalised inverse
for positive operator valued functions can be made.

\begin{theorem} The following statements are equivalent for $\nu_1,\nu_2\in\povm_\hil(X)$:
\begin{enumerate}
\item $\nu_2 \ll_{\rm ac} \nu_1$;
\item there exists a bounded Borel function $g:X\rightarrow\bh$, unique up to sets of $\nu_1$-measure zero, such that
\begin{equation}\label{rn prop}
\nu_2(E)\,=\,\int_E g\,d\nu_1 ,\;\mbox{ for every }E\in\bor(X)\,.
\end{equation}
\end{enumerate}
If the equivalent conditions above hold and if $\mu_j$ is the probability measure induced by $\nu_j$, then
$\mu_2\ll_{\rm ac} \mu_1$ and
\begin{equation}\label{rn deriv}
g\,=\,\left(\frac{d\mu_2}{d\mu_1}\right)
\left[ 
\left(\frac{d\nu_1}{d\mu_1}\right)^{-1/2}\left(\frac{d\nu_2}{d\mu_2}\right)\left(\frac{d\nu_1}{d\mu_1}\right)^{-1/2}
\right]\,.
\end{equation}
\end{theorem}

\begin{proof}
Assume that $\nu_2 \ll_{\rm ac} \nu_1$. If $\mu_1(E)= \frac{1}{d}\tr(\nu_1(E))=0$, then $\nu_1(E)=0$. 
By assumption $\nu_2(E)=0$ and therefore $\mu_2(E)= \frac{1}{d}\tr(\nu_2(E))=0$, which proves that $\mu_2 \ll_{\rm ac} \mu_1$.
Therefore, for any $E\in\bor(X)$, we have
\[
\nu_2^{(i,j)}(E)=\langle \nu_2(E) e_j, e_i\rangle \ll_{\rm ac} \langle \nu_1(E) e_j, e_i \rangle = \nu_1^{(i,j)}(E).
\]
Coordinate-wise we obtain $\nu_2^{(i,j)} \ll_{\rm ac} \mu_1$.
By applying the chain rule for the classical Radon-Nikod\'ym derivatives we obtain
\[
\displaystyle \frac{d\nu_2^{(i,j)}}{d\mu_1} = \displaystyle \frac{d\nu_2^{(i,j)}}{d\mu_2} \displaystyle \frac{d\mu_2}{d\mu_1}
\]
Hence,
\[
\displaystyle \frac{d\nu_2}{d\mu_1} = \displaystyle \frac{d\nu_2}{d\mu_2} \displaystyle \frac{d\mu_2}{d\mu_1},
\]
where $\displaystyle\frac{d\nu_2}{d\mu_2}:X\rightarrow\bh_+$ and $\displaystyle \frac{d\mu_2}{d\mu_1}:X\rightarrow\mathbb R_+$.
With $g$ as above,
\[
\left( \frac{d\nu_1}{d\mu_1}\right)^{1/2} g  \left( \frac{d\nu_1}{d\mu_1}\right)^{1/2} \,=\,
\left(\frac{d\mu_2}{d\mu_1}\right)\,\frac{d\nu_2}{d\mu_2}\,.
\]
Thus, for any state $\rho$ and $E\in\bor(X)$,
\[
\begin{array}{rcl}
\tr\left ( \rho \displaystyle \int_E g d\nu_1 \right) &=& \displaystyle \int_E \tr \left ( \rho \left(\frac{d\nu_1}{d\mu_1}\right)^{1/2}g 
\left(\frac{d\nu_1}{d\mu_1}\right)^{1/2} \right) d\mu_1
\\ && \\
&=& \displaystyle\int_E  \left[\tr \left( \rho \frac{d\nu_2}{d\mu_2} \right) \right]\frac{d\mu_2}{d\mu_1}\, d\mu_1
\\ && \\
&=& \displaystyle \int_E \tr \left( \rho \frac{d\nu_2}{d\mu_2}\right)d\mu_2
\\ && \\
&=& \tr (\rho \, \nu_2(E)).
\end{array}
\]
Therefore, by definition of the integral, $\nu_2(E)=\displaystyle\int_E g \, d\nu_1$ for every $E\in\bor(X)$.
Passing to the $d^2$ coordinate measures and using the uniqueness of the classical Radon-Nikod\'ym derivative,
one deduces that $g$ is unique up to sets of $\nu_1$-measure zero.

Conversely, assume such a function $g:X\rightarrow\bh$ exists such that
\[
%\begin{equation}\label{rn prop}
\nu_2(E)\,=\,\int_E g\,d\nu_1 ,\;\mbox{ for every }E\in\bor(X)\,.
%\end{equation}
\]
If $\nu_1(E)=0$, then $\nu_2(E)=\int_E g\,d\nu_1=0$ and thus $\nu_2 \ll_{\rm ac} \nu_1$.
\end{proof}

The function $g$ in \eqref{rn deriv} is called a \emph{non-principal Radon-Nikod\'ym derivative} of $\nu_2$ with respect to $\nu_1$.

%%%%%%%%%%%%%%%%%%%%%%%%%%%%%%%%%%%%%%%%%%%%%%%%%%%%%%
\section{Topology of $\povm_\hil(X)$}

In classical probability, $P(X)$ is a weak*-closed subset of the unit sphere of the dual space of the
abelian C$^*$-algebra $C(X)$. Hence, by the Banach-Alaoglu Theorem, $P(X)$ is compact. In this section
we introduce an analogous topology on $\povm_\hil(X)$ so that it not only is a compact topological space, but also 
has the property that the POVMs with finite support are dense in  $\povm_\hil(X)$.

\begin{definition}
Define $\Omega_0:\elem_\hil(X) \rightarrow \povm_\hil(X)$ by
\[
\Omega_0\left( \sum_{j=1}^n t_j^*\varrho_{x_j}t_j \right) = \sum_{j=1}^n \delta_{x_j}t_j^* t_j\,.
\]
\end{definition}

\begin{proposition}
The map $\Omega_0$ is $C^*$-affine.
\end{proposition}

\begin{proof} If $\phi=\displaystyle\sum_{j=1}^n t_j^*\varrho_{x_j}t_j$, then
$\Omega_0(\phi)=\sum_{j=1}^n \delta_{x_j} t_j^*t_j$.  Let $a\in\bh$ and note that
\[
\Omega_0(a^* \phi \, a)=\sum_{j=1}^n \delta_{x_j}a^*t_j^*t_ja=a^*\Omega_0(\phi)a \,.
\]
Thus, if $\phi_1,...,\phi_m\in\elem_\hil$ and if $a_1,...,a_m\in\bh$ are C$^*$-convex coefficients, then
\[
\Omega_0\left( \sum_{k=1}^m a_k^* \phi_k \, a_k \right) = \sum_{k=1}^m a_k^* \Omega_0(\phi_k) \, a_k \in \povm_\hil(X) \, ,
\]
which completes the proof.
\end{proof}

We aim to show that $\Omega_0$ admits a C$^*$-affine extension to $\ucp_\hil(X)$. 
Before discussing the extension map, we will describe the topologies of $\ucp_\hil(X)$ and $\povm_\hil(X)$.
Because topologies can be characterised by how nets converge, we begin with the BW-topology of $\ucp_\hil(X)$
and use $\Gamma$ to induce a topology on $\povm_\hil(X)$.

\begin{definition} {\rm (The bounded-weak topology)}
A net $\{\phi_\gamma\}_{\gamma\in\Lambda} \subset \ucp_\hil(X)$ converges to $\phi \in \ucp_\hil(X)$ 
if $\displaystyle\lim_\gamma \left\|\phi_\gamma (f) - \phi (f) \right\| = 0$ for all $f\in C(X)\otimes\bh$.
\end{definition}

In the BW-topology, $\ucp_\hil(X)$ is a compact space  \cite{arveson1969,Paulsen-book}.
We use the topology of $\ucp_\hil(X)$ to induce a topology on $\povm_\hil(X)$ as follows.

\begin{definition}
A net $\{\nu_\gamma\}_{\gamma\in\Lambda} \subset \povm_\hil(X)$ converges to $\nu\in\povm_\hil(X)$ 
if
\[
\int_X f\,d\nu\,=\,\lim_\gamma\,\int_X f\,d\nu_\gamma\;\mbox{ for every } f\in C(X)\otimes\bh\,.
\]
That is, $\nu_\gamma\rightarrow\nu$ 
if the net $\{\Gamma(\nu_\gamma)\}_{\gamma\in\Lambda}$ converges to $\Gamma(\nu)$ in $\ucp_\hil$.
\end{definition}
 
\begin{theorem}\label{top} There exists a $C^*$-affine function $\Omega:\ucp_\hil(X)\rightarrow\povm_\hil(X)$
such that 
\begin{enumerate}
\item $\Omega_{\vert \elem_\hil(X)}=\Omega_0$ and
\item $\Omega\circ\Gamma={\rm id}_{\povm_\hil(X)}$.
\end{enumerate}
\end{theorem}

\begin{proof} We first prove that $\elem_\hil(X)$ is BW-dense in $\ucp_\hil(X)$.
Because $\ucp_\hil$ is BW-compact and convex, the
Kre\v{\i}n-Milman theorem asserts that the convex hull of the extreme points of $\ucp_\hil(X)$
is a dense subset of $\ucp_\hil(X)$.
If $\phi\in\ext\,(\ucp_\hil(X))$, then there exist irreducible representations $\pi_j:C(X)\otimes\bh\rightarrow\bof(\hil_{\pi_j})$
and an isometry $v:\hil\rightarrow\hil_\pi=\bigoplus_{j=1}^m\hil_{\pi_j}$ such that $\phi=v^*\pi v$ and the linear map
$y\mapsto v^*yv$ is one-to-one on the commutant $\mathcal N$ of $\pi\left(C(X)\otimes\bh\right)$
\cite[Theorem 1.4.6]{arveson1969}.
Because $\dim\hil=d<\infty$, every irreducible representation $\pi_j$ of $C(X)\otimes\bh$ takes place on a Hilbert space $\hil_{\pi_j}$ of dimension $d$ and 
has the form $\pi=u_j^*\varrho_{x_j}u_j$, where $\varrho_{x_j}\in\ucp_\hil(X)$ is spectral and $u_j:\hil_{\pi_j}\rightarrow\hil$ is unitary. Because
$\hil_{\pi_j}$ and $\hil$ are of dimension $d$, we may assume without loss of generality that $\hil_{\pi_j}=\hil$ for every $j$, whence
$\pi_j=\alpha_j\circ\varrho_{x_j}$ for the automorphism $\alpha_j\in\aut(\bh)$ given by $\alpha_j(z)=u_j^*zu_j$.
If $q_j\in\bof(\hil_\pi)$ is the projection of $\hil_\pi$ onto the direct summand $\hil_{\pi_j}=\hil$, then with $a_i=u_iq_iv\in\bh$ we obtain
\[
\phi\,=\,\sum_{j=1}^m a_j^*\varrho_{x_j}a_j\,,
\]
which is an element of $\elem_\hil(X)$. Therefore, $\elem_\hil(X)$ contains the extreme points of $\ucp_\hil(X)$. By the 
convexity of $\elem_\hil(X)$ and the Kre\v{\i}n-Milman theorem we deduce that $\elem_\hil(X)$ is BW-dense in $\ucp_\hil(X)$.

Now assume that $\phi \in \ucp_\hil(X)\setminus \elem_\hil(X)$.  
Thus, there exists a net $\{\phi_\gamma\}_{\gamma\in\Lambda} \subset\elem_\hil$ such 
that $\|\phi_\gamma(f)-\phi(f)\|\rightarrow 0$ for all $f\in C(X)\otimes\bh$.  For each pair $(i,j)$ let $(\phi_\gamma)_{ij}\in C(X)^*$ be given by
\[
(\phi_\gamma)_{ij} (g) = \langle \phi_\gamma (g\otimes e_{ij}) \, e_j, e_i \rangle
\]
for all $g\in C(X)$. Thus, 
\[
\phi_\gamma = \displaystyle \sum_{i,j =1}^d (\phi_\gamma)_{ij}\otimes e_{ij}.
\]
By the classical
Riesz Representation Theorem, there exists a complex measure $\nu_{\gamma_{ i j}}$ on 
the Borel sets of $X$ such that 
\[
(\phi_\gamma)_{ij}(g) = \displaystyle \int_X g \,d\nu_{\gamma_{ i j}} \; \mbox{for all } g\in C(X).
\]
Since $\phi_\gamma(f)\rightarrow \phi(f) \;\mbox{for all }f\in C(X)\otimes\bh$
 we have that $(\phi_\gamma)_{ij}(g)\rightarrow \phi_{ij}(g)$ for all $g\in C(X)$ and pairs $(i,j)$.  
Therefore $(\nu_{\gamma})_{ij} \rightarrow \nu_{ij}$.  
Define 
\[
\nu = \displaystyle \sum_{i,j=1}^d \nu_{ij}\otimes e_{ij} \in \povm_\hil(X)
\]
and let $\Omega(\phi)=\nu$.  This is well defined as $\nu$ is independent of the choice of approximating net $\{\phi_\gamma\}$.

Let $\phi_1, ..., \phi_m\in \ucp_\hil(X)$ and suppose that $\{\phi_{\gamma_j}\}_{\gamma_j \in\Lambda_j} \subset \elem_\hil(X)$ is an 
approximating net for $\phi_j$, for each
$j$. Let $\Lambda=\Lambda_1\times \cdots\times\Lambda_m$, ordered as a directed set in the natural way. 
If $t_1,...,t_m\in\bh$ are C$^*$-convex coefficients, %and if $\{\phi_{\gamma}\}_{\gamma_j \in\Lambda_j} \subset \elem_\hil(X)$ is an 
%approximating net for $\phi_j$, 
then with respect to the natural induced directed-set ordering on $\Lambda_1\times \cdots\times\Lambda_m$, the ucp maps
\[
\left\{ \displaystyle \sum_{j=1}^m t_j^*\phi_{\gamma_j} t_j \right\}_{\gamma=(\gamma_1,...,\gamma_m)\in\Lambda}
\]
form an approximating net for $\displaystyle \sum_{j=1}^m t_j^*\phi_j t_j$.   
Because $\Omega_0$ is C$^*$-affine,
\[
\begin{array}{rcl}
\Omega \left(\displaystyle \sum_{j=1}^m t_j^*\phi_j t_j \right) &=& 
\displaystyle \lim_\gamma \Omega_0 \left( \displaystyle \sum_{j=1}^m t_j^*\phi_{\gamma_j} t_j \right)
\\ && \\
&=& \displaystyle \lim_\gamma \displaystyle \sum_{j=1}^m t_j^* \Omega_0(\phi_{\gamma_j}) t_j
\\ && \\
&=& \displaystyle \sum_{j=1}^m t_j^* \Omega(\phi_j) t_j\,,
\end{array}
\]
which establishes the extension $\Omega$ of $\Omega_0$.

Observe that if $\nu=\displaystyle\sum_{j=1}^n \delta_{x_j}h_j$, then $\Omega\circ\Gamma(\nu)=\nu$. Suppose now that $\nu\in\povm_\hil(X)$
is arbitrary and 
let $\{\phi_\gamma\}_\gamma$ be a net in $\elem_\hil(X)$ BW-convergent to $\Gamma(\nu)$. For each $\gamma$, let $\nu_\gamma=\Omega\phi_\gamma$
so that $\Omega\Gamma(\nu_\gamma)=\nu_\gamma$. As shown earlier, another realisation of $\nu$ is via
\[
\nu\,=\,\sum_{i,j=1}^d \nu_{ij}\otimes e_{ij}\,,
\]
where $\nu_{ij}$ is attained in the dual space of $C(X)$ as a the limit of the net of measures induced by the linear functionals
$(\phi_\gamma)_{ij}(g)=\langle \phi(g\otimes e_{ij})e_j,e_i\rangle$. Thus
$\int_X f\,d\nu_\gamma\rightarrow\int_X f\,d\nu$ for all $f\in C(X)\otimes \bh$ and $\Omega\Gamma(\nu)=\nu$.
 \end{proof}

\begin{corollary}
$\povm_\hil(X)$ is a compact space.
\end{corollary}

\begin{proof} If $\{\nu\}_\gamma$ is a net in $\povm_\hil(X)$, then $\{\Gamma \nu_\gamma\}_{\gamma}$ is a net in $\ucp_\hil(X)$. 
Because $\ucp_\hil(X)$ is compact, there is a subnet $\{\Gamma \nu_{\gamma_j}\}_{j}$ which is convergent, say to $\phi$.
Thus, if $\nu=\Omega\phi$, then $\{\nu_{\gamma_j}\}_{j}$ converges to $\nu$. Hence, every net in $\povm_\hil(X)$
admits a convergent subnet, and so $\povm_\hil(X)$ is compact.
\end{proof}

A consequence of the Kre\v{\i}n-Milman Theorem and the proof of Theorem \ref{top} is:

\begin{corollary}\label{approx} The set of all quantum measurements with finite support is dense in $\povm_\hil(X)$.
\end{corollary}

Operators $z$ acting on finite-dimensional Hilbert spaces $\hilk$ admit polar decompositions of the form $z=u|z|$, 
where $|z|=(z^*z)^{1/2}$ and $u\in\bof(\hilk)$ is unitary. (The unitary need not be unique if $z$ is not invertible.) Therefore, 
\[
\begin{array}{rcl}
\Gamma\circ\Omega_0 \left( \displaystyle \sum_{j=1}^n t_j^*\varrho_{x_j} t_j \right) &=& \Gamma \left( \displaystyle\sum_{j=1}^n \delta_{x_j}t_j^*t_j \right)
\\ && \\
&=& \displaystyle\sum_{j=1}^n (t_j^*t_j)^{1/2}\varrho_{x_j}(t_j^*t_j)^{1/2}
\\ && \\
&=& \displaystyle\sum_{j=1}^n |t_j|\, \left(\alpha_j\circ\varrho_{x_j} \right)|t_j|,
\end{array}
\]
where $t_j=u_j\,|t_j|$ is a polar decomposition of $t_j$ and $\alpha_j\in\aut(\bh)$ is given by $\alpha_j(z)=u_j^*zu_j$. 
Thus, if $\phi\in\epos_\hil(X)$, we have $\Gamma\circ\Omega_0(\phi)=\phi$.

Let $\overline{\epos_\hil(X)}$ denote the BW-closure of $\epos_\hil(X)$ in $\ucp_\hil(X)$.

\begin{corollary} {\rm (Riesz Representation Theorem)} 
For every $\phi\in\overline{\epos_\hil(X)}$ there is a
unique $\nu\in\povm_\hil(X)$
such that
\[
\phi(f)\,=\,\int_X f\,d\nu\,\;\mbox{ for every }f\in C(X)\otimes\bh\,.
\]
\end{corollary}

\begin{proof} Let $\nu=\Omega\phi$ so that $\Gamma\nu=\Gamma\Omega\phi=\phi$. If $\nu'$ is another POVM 
for which $\Gamma(\nu')=\phi$, then $\nu'=\Omega\Gamma\nu'=\Omega\phi=\nu$.
\end{proof}

%%%%%%%%%%%
\section{Classical Randomness}

In this section we establish the following result.
 
 \begin{theorem}\label{ext pt}
The following statements are equivalent for $\nu\in\povm_\hil(X)$:
\begin{enumerate}
\item $\nu$ is an extreme point of $\povm_\hil(X)$;
\item there exist distinct $x_1,\dots, x_m\in X$ and $h_1,\dots,h_m\in\bh_+$ such that
    \begin{enumerate}
        \item the subspaces $\mbox{\rm ran}\,h_1,\dots, \mbox{\rm ran}\,h_m$  are weakly independent and
        \item $\nu\,=\,\displaystyle\sum_{j=1}^m\delta_{x_j}h_j$.
    \end{enumerate}
\end{enumerate}
\end{theorem}

In the case of finite $X$, Theorem \ref{ext pt} is already known \cite{dariano--etal2005,parthasarathy1999}. Our contribution is
to show that the case of arbitrary $X$ reduces to the case of finite $X$ (Lemmas \ref{res-support} and \ref{fin-support}). 
However, for completeness, we include a full proof
of Theorem \ref{ext pt},
adapting the elegant arguments of D'ariano, Lo Presti, and Perinotti \cite{dariano--etal2005}.  
The concept of weak independence is defined formally below.

\begin{definition}
Subspaces $\mathcal L_1,...,\mathcal L_n\subset\hil$ are weakly independent 
if, for any $t_1,...,t_n\in\bh$ such that (i) ${\rm ran}\,t_j + {\rm ran}\,t_j^*\subset\mathcal L_j$ for each $j$, and (ii)
$t_1+...+t_n=0$, then necessarily each $t_j=0$.
\end{definition}

The remainder of this section is devoted to the proof of Theorem \ref{ext pt} and a discussion of
some of its consequences.

Recall that
if $K\subset X$ is a closed subset, then the $\sigma$-algebra $\bor(K)$ of Borel sets of $K$
is given by $\mathcal O(K)=\{K\cap E\,:\,E\in\bor(X)\}$.

\begin{lemma}\label{res-support} Assume $K_\nu\subset X$ is the support of $\nu\in\povm_\hil(X)$. Then
$\nu\in\ext(\povm_\hil(X))$ if and only if the restriction $\nu_{\vert \bor(K_\nu)}$ of
$\nu$ to $\bor(K_\nu)$ is an extreme point of $\povm_\hil(K_\nu)$.
\end{lemma}

\begin{proof}
Assume that $\nu\in\ext(\povm_\hil(X))$.  Let $\nu_0,\nu_1,\nu_2\in\povm_\hil(K_\nu)$ and 
such that $\nu_{\vert_{\bor(K_\nu)}}=\nu_0=\displaystyle\frac{1}{2}(\nu_1+\nu_2)$.
Define $\tilde\nu_j:\bor(X)\rightarrow\bh$ by $\tilde\nu_j(E)=\nu_j(E\cap K_\nu)$ for all 
$E\in\bor(X)$, to obtain $\tilde\nu_j\in\povm_\hil(X)$. Because $K_\nu$ is the support 
of $\nu$, $\nu(E)=\nu(E\cap K_\nu)$ for all $E\in\bor(X)$; 
thus, $\nu=\frac{1}{2}(\tilde\nu_1 +\tilde\nu_2)$, and so $\nu=\tilde\nu_1=\tilde\nu_2$, and so $\nu_0=\nu_1=\nu_2$.

Conversely, assume that $\nu_0=\nu_{\vert_{\bor(K_\nu)}}$ is an extreme point of $\povm_\hil(K_\nu)$.  
Let $\nu=\frac{1}{2}(\nu_1+\nu_2)$ for $\nu_1,\nu_2\in\povm_\hil(X)$.  If $E\in\bor(X)$ satisfies
$\nu(E)=0$ then $0=\nu(E)\geq  \frac{1}{2}\nu_j\geq 0$ implies that $\nu_j=0$.  Thus, $\nu_j\ll_{\rm ac}\nu$. 
If we show that the support of each $\nu_j$ is contained in the support of $\nu$, then we conclude that $\nu_1=\nu_2=\nu$.

Thus, it remains to prove that if
$\omega,\nu\in\povm_\hil(X)$ is such that $\omega\ll_{\rm ac}\nu$, then
$K_\omega\subset K_\nu$.
To this end, let $U=(X\setminus K_\nu)\cap(X\setminus K_\omega)$, which is open, and let $K=(X\setminus U)\cap K_\omega$, which is closed. Thus,
\[
\omega(X\setminus K) = \omega\left(U\cup (X\setminus K_\omega)\right) 
\leq \omega(U) + \omega(X\setminus K_\omega)
= \omega(U).
\]
Now since $U\subset X\setminus K_\nu$, we have $\nu(U)\leq\nu(X\setminus K_\nu)=0$.  Thus, 
$\omega\ll_{\rm ac}\nu$ implies that $\omega(U)=0$ and so $\omega(X\setminus K)=0$.
Hence, $K\subset K_\omega$ and $\omega(X\setminus K)=0$ which implies that $K=K_\omega$ by definition of support
and by the above arguments.  
Hence, $K_\omega = K = (X\setminus U)\cap K_\omega = (K_\nu\cup K_\omega)\cap K_\omega$ implies that $K_\omega\subset K_\nu$.
\end{proof}

\begin{lemma}\label{fin-support}
If $\nu\in \povm_\hil(X)$ is an extreme point, then the support of $\nu$ is a finite set.
\end{lemma}

\begin{proof} Assume, contrary to what we aim to prove, that the support of $\nu$ is an infinite set.
By Lemma \ref{res-support}, we may replace $X$ with the support of
$\nu$, and so we assume without loss of generality that $X=K_\nu$.
The argument below is inspired by the proof of the main result of \cite{douglas1964}.

Let $\mu=\frac{1}{d}\tr\circ\nu$ and consider $L^1(X,\mu)$.  
Because $\mu$ and $\nu$ are mutually absolutely continuous, they have the same support $X$.  
Thus $L^1(X,\mu)$ is an infinite-dimensional Banach space.

Let $\kappa_{ij}(x) = \left\langle \left( \displaystyle\frac{d\nu}{d\mu}\right)\, e_j, \,e_i \right\rangle$, for $1\leq i,j,\leq d$, and let
\[
\begin{array}{rcl}
Q_0 &=& \left\{ \tr\left[\left( \displaystyle\frac{d\nu}{d\mu}\right)^{1/2}\rho\left(\displaystyle\frac{d\nu}{d\mu}\right)^{1/2}\right]\,:\,\rho\in\state(\hil)\right\} 
\\ && \\
&=& \left\{ \tr\left(  \rho \, \displaystyle\frac{d\nu}{d\mu}\right)\,:\,\rho\in\state(\hil)\right\} 
\\ && \\
& \subset & \mbox{Span}\,\{\kappa_{ij}\,:\,1\leq i,j\leq d\}\,.
\end{array}
\]
Let $Q=\mbox{Span}\,Q_0$; thus, $Q$ is a subspace of $L^1(X,\mu)$ of dimension at most $d^2$. Because $L^1(X,\mu)$ has infinite dimension,
the annihilator of $Q$ in the dual space $L^1(X,\mu)^*=L^\infty(X,\mu)$ has infinite dimension.  
Hence, there exists $\varphi\in L^\infty(X,\mu)$ such that $\|\varphi\|=1$ and 
$\displaystyle\int_X \varphi\,\psi\, d\mu = 0$ for all $\psi\in Q$. Let $\Phi=\sum_{j=1}^d \varphi\otimes e_{jj}$; thus,  
for any state $\rho \in \state(\hil)$,
\[
\int_X \tr \left[ \left(\frac{d\nu}{d\mu}\right)^{1/2} \rho  \left(\frac{d\nu}{d\mu}\right)^{1/2} \Phi \right] \, d\mu 
\,=\, \int_X \tr\left(\varphi \, \rho \, \displaystyle\frac{d\nu}{d\mu}\right)\,d\mu
\,=\, 0\,.
\]
Hence,
\begin{equation}
\label{zeroint} \int_X \Phi \, d\nu = 0.
\end{equation}

Define $\tilde\nu: \bor(X)\rightarrow \bh$ by $\tilde\nu(E) = \int_E\Phi\, d\nu$. Let $\nu_1=\nu+\tilde\nu$ and $\nu_2=\nu-\tilde\nu$.  Note that
\[
\nu_1(E)=\int_E d(\nu+\tilde\nu) = \int_E d\nu +    \int_E\Phi\, d\nu =  \int_E \left(1 + \Phi \right)\, d\nu.
\]
Since $1 + \Phi$ is positive for all $x\in X$, this final integral above is a positive operator. Likewise $\nu_2(E)$ is positive. Further, by equation \eqref{zeroint},
\[
\nu_1(X)=\int_X d(\nu+\tilde\nu) = \nu(X) +   \int_X\Phi\, d\nu =  \nu(X) + 0 = 1.
\]
Hence, $\nu_1,\nu_2\in\povm_\hil(X)$ and $\nu=\frac{1}{2}\nu_1+\frac{1}{2}\nu_2$.  Since $\nu_1\neq\nu$ (because $\varphi\neq 0$), we have $\nu\notin \ext(\povm_\hil(X))$.
\end{proof}

We are now prepared for the proof of Theorem \ref{ext pt}, using an adaptation of the
arguments of the D'ariano, Lo Presti, and Perinotti  \cite{dariano--etal2005}. Recall that we aim to prove that 
the following statements are equivalent for $\nu\in\povm_\hil(X)$:
\begin{enumerate}
\item\label{ext-1} $\nu$ is an extreme point of $\povm_\hil(X)$;
\item\label{ext-2} there exist distinct $x_1,\dots, x_n\in X$ and $h_1,\dots,h_n\in\bh_+$ such that
    \begin{enumerate}
        \item the subspaces $\mbox{\rm ran}\,h_1,\dots, \mbox{\rm ran}\,h_n$  are weakly independent and
        \item $\nu\,=\,\displaystyle\sum_{j=1}^n\delta_{x_j}h_j$.
    \end{enumerate}
\end{enumerate}

\begin{proof} Assume \eqref{ext-1}. Thus, $\nu\in\ext\left(\povm_\hil(X)\right)$. 
By Lemma \ref{res-support}, we may replace $X$ with the support $K_\nu$ of
$\nu$, which by Lemma \ref{fin-support} is a finite set if $\nu\in\ext\left(\povm_\hil(X)\right)$. Thus,
without loss generality we may assume that $X$ is a finite set, say $X=\{x_1,\dots,x_n\}$, and that the support of $\nu$ is $X$. Hence,
$\nu=\displaystyle\sum_{j=1}^n\delta_{x_j}h_j$. 
Suppose that $t_1,...,t_n\in\bh$ satisfy $\displaystyle\sum_{j=1}^m t_j=0$
and ${\rm ran}\,t_j+{\rm ran}\,t_j^*  \subset{\rm ran}\,h_j $ for all $1\leq j \leq n$. Also assume, contrary to what we aim to prove, that not every $t_j$ is
zero. If every $t_j$ is hermitian, then let $g_j=t_j$ for all $j$ and $g_k\neq0$ for some $k$;  if not all $t_1,...,t_n$ are hermitian, then there exists 
$k$ with $\Im(t_k)\neq 0$ and in this situation we take $g_j=\Im(t_j)$ for all $1\leq j \leq n$.  
(The imaginary part  of $s\in \bh$ is the hermitian operator $\Im (s)=\frac{1}{2i}(s-s^*)$.)
With this choice of operators $g_j$, we have ${\rm ran} g_j\subset{\rm ran}\,h_j$ for all $1\leq j \leq n$, $g_k\neq 0$ for some $k$,
and $\displaystyle\sum_{j=1}^n g_j =0$.
For each $j$ write $\hil=\ker h_j\oplus {\rm ran}\,h_j$ so that 
\[
h_j\,=\,\left[ \begin{array}{cc} 0 & 0 \\ 0 & \tilde h_j  \end{array}\right] \quad\mbox{ and }\quad
g_j\,=\, \left[ \begin{array}{cc} 0& 0 \\ 0 & \tilde g_j  \end{array}\right] \,,
\]
where $\tilde h_j \in \bof({\rm ran}\,h_j)$ is positive and invertible.  (If $h_j$ itself is invertible, then we do not use a $2\times 2$ operator
matrix and simply take $\tilde h_j=h_j$.) Let
\[
\mathcal Z_\varepsilon\,=\,\left\{ 1\pm\epsilon\lambda \,\vert\, \lambda\in \displaystyle\bigcup_{j=1}^n \sigma\left( \tilde h_j^{-1/2}\tilde g_j \tilde h_j^{-1/2}\right) \right\}\,,
\]
where $\sigma(z)$ denotes the spectrum of an operator $z$. Because $\mathcal Z_\varepsilon$ is a finite set,
there exists $\varepsilon>0$ so that $\mathcal Z_\varepsilon \subset (0,\infty)$.  Hence, 
\[
h_j\pm \varepsilon g_j\,=\,
\left[\begin{array}{ccc} 0&& 0 \\ 0 &&
\tilde h_j^{1/2}(1\pm\varepsilon \tilde h_j^{-1/2}\tilde g_j \tilde h_j^{-1/2})\tilde h_j^{1/2} \end{array}\right]\,,
\]
which is positive by the choice of $\varepsilon>0$.
Now let $\gamma = \displaystyle \sum_{j=1}^n \delta_{x_j} g_j$ to obtain
$\nu\pm\varepsilon\gamma\in\povm_\hil(X)$. If $\nu_1=\nu+\varepsilon\gamma$ and $\nu_2=\nu-\varepsilon\gamma$, then
$\nu_1,\nu_2$ are distinct elements of $\povm_\hil(X)$ 
and $\nu$ is their midpoint, in contradiction to the hypothesis $\nu\in\ext\left(\povm_\hil(X)\right)$. 
Hence, it must be that all of the operators $t_j$ are zero, which is to say the ranges of $h_1,\dots,h_n$ are weakly independent.

Conversely, assume \eqref{ext-2}. 
Thus, $\nu=\displaystyle\sum_{j=1}^n\delta_{x_j}h_j$ for distinct $x_1,\dots, x_n\in X$
and positive operators $h_1,\dots, h_n\in\bh$ with weakly independent ranges. The support of $\nu$ is
given by $K_\nu=\{x_1,\dots,x_n\}$.
By Lemma \ref{res-support}, $\nu$ is an extreme point of $\povm_\hil(X)$ if and only if $\nu$ is
an extreme point of $\povm_\hil(K_\nu)$. Hence, we assume without loss of generality that
$X=K_\nu$.
If, contrary to what we wish to prove, $\nu\not\in\ext\left(\povm_\hil(X)\right)$, then
in the real vector space $W$ of 
countably additive functions $\upsilon:\bor(X)\rightarrow\bh_{\rm sa}$ 
there exists
 $\varepsilon>0$ and $\gamma\in W$ such that $\omega^\pm = \nu \pm \varepsilon\gamma \in \povm_\hil(X)$. Therefore, 
there exist $g_1,...,g_n\in\bh_{\rm sa}$ such that $\gamma=\displaystyle\sum_{j=1}^n\delta_{x_j}g_j$ and $g_k\neq 0$ for at least one $k$.  Because 
\[
1=\omega^\pm(X)=\nu(X)\pm \varepsilon\gamma(X)=1\pm\varepsilon \displaystyle\sum_{j=1}^n g_j,
\]
we obtain $\displaystyle \sum_{j=1}^n g_j =0$.  And because $\omega^\pm(E)\in\eff(\hil)$ we obtain through evaluation at $x_j$
that $ h_j\pm \epsilon g_j\in\bh_+$.
Now write $\hil=\ker h_j\oplus {\rm ran}\,h_j$ so that 
\[
h_j\,=\,\left[ \begin{array}{cc} 0 & 0 \\ 0 & \tilde h_j  \end{array}\right] \quad\mbox{ and }\quad
g_j\,=\, \left[ \begin{array}{cc} a_j & y_j \\ y_j^* & \tilde g_j  \end{array}\right] \,,
\]
where $a_j\in \bof\left(\ker h_j\right)$ is hermitian and $\tilde h_j \in \bof({\rm ran}\,h_j)$ is positive and invertible.  
Because $h_j+\varepsilon g_j$ is positive, $a_j$ is necessarily positive. But $h_j-\varepsilon g_j$ positive implies that $-a_j$ is positive.
Hence $a_j=0$. The positivity of $h_j+\varepsilon g_j=\left[\begin{array}{cc} 0 & \varepsilon y_j \\ \varepsilon y_j^* & \varepsilon\tilde g_j\end{array}\right]$
yields $y_j=0$, and so $g_j\xi=0$ for every vector $\xi\in\hil$ for which $h_j\xi=0$.
That is,
$\ker h_j\subset\ker g_j$ and so $ {\rm ran}\, g_j\subset{\rm ran}\, h_j$.  
We conclude that ${\rm ran}\,h_1,\dots, {\rm ran}\, h_n$ are not weakly independent, contrary to hypothesis. Hence, it must be that no such
function $\gamma$ exists, which is to say that $\nu$ is an extreme point of $\povm_\hil(X)$.
\end{proof}

\begin{corollary} If $\nu$ is an extreme point of $\povm_\hil(X)$, then so is $\alpha\circ\nu$ for every automorphism $\alpha$ of $\bh$.
\end{corollary}

\begin{proof} One need only note that if $u\in\bh$ is unitary and $h_1,\dots,h_n\in\bh$ are positive, then $h_1,\dots, h_n$ have
weakly independent ranges if and only if $u^*h_1u$, \dots, $u^*h_nu$ have weakly independent ranges.
\end{proof}

Although the function $\Gamma$ does not exhibit affine properties, it does map extremal elements to extremal elements.

\begin{corollary}\label{ext cor} If $\nu$ is an extreme point of $\povm_\hil(X)$, then $\Gamma(\nu)$ is an extreme point of $\ucp_\hil(X)$.
\end{corollary}

\begin{proof} Assume that $\nu\in\ext\left(\povm_\hil(X)\right)$ and that $\phi_\nu=\Gamma(\nu)$. Hence, by Theorem \ref{ext pt},
$\phi_\nu= \sum_{j=1}^n h_j^{1/2}\varrho_{x_j} h_j^{1/2}$
for some distinct $x_1,\dots,x_n\in X$ and positive operators $h_1,\dots,h_n$ with weakly independent
ranges.
Each $\varrho_{x_j}$ is an irreducible representation of $C(X)\otimes\bh$ on $\hil$ and,
because the points $x_1,\dots, x_n\in X$ are distinct, no two irreducible representations $\varrho_{x_i}$
and $\varrho_{x_j}$ corresponding to distinct points $x_i$ and $x_j$ are unitarily equivalent. By weak independence, every
$h_j\neq0$; and by irreducibility each completely positive map $\phi_j:C(X)\otimes\bh\rightarrow\bh$ of the form
\[
\phi_j\,=\,h_j^{1/2}\varrho_{x_j}h_j^{1/2}
\]
generates an extremal ray of the cone of all completely positive linear maps $C(X)\otimes\bh\rightarrow\bh$ \cite[Corollary 1.4.3]{arveson1969}.
Therefore, by \cite[Lemma 1.4.9]{arveson1969}, $\phi_\nu=\sum_{j=1}^n\phi_j$ is an extremal ucp map.  
 \end{proof}

%%%%%%%%%%%%%
\section{Nonclassical Randomness}

 A convex combination $\sum_{j=1}^m\lambda_j\nu_j$ of $\nu_1,\dots,\nu_m\in\povm_\hil(X)$ is a
 random POVM if one views the set $\{\lambda_1,\dots,\lambda_m\}$ of convex coefficients as a probability distribution.
 The notion of classical randomness enters quantum measurement through probabilistic (that is, convex combinations)
 mixtures of other quantum measurements.

 In nonclassical convexity, there are corresponding notions of randomness
 afforded by C$^*$-convex coefficients. We mention below two sources of randomness.
 
 First, assume that $a=(a_1,\dots,a_m)$ is a tuple of C$^*$-convex coefficients $a_j\in\bh$.
 Because $\hil$ has finite dimension, the subspaces $\ker a_j$ and $\ker a_j^*$ have equal dimension for each $j$.
 Hence, the isometry $v:\hil\rightarrow\bigoplus_{j=1}^m \hil$ defined by
 \[
 v\xi\,=\,\bigoplus a_j\xi\,,\quad\xi\in \hil\,,
 \]
 extends to a unitary $u\in\bof\left(\bigoplus_{j=1}^m\hil\right)$ \cite[Corollary 2.2]{araki--hansen2000}.
 Conversely, expressing any unitary $u$ acting on $\bigoplus_{j=1}^m \hil$ as an $m\times m$
matrix of operators and by selecting the operators that appear in any single column of $u$, one obtains
a tuple of C$^*$-convex coefficients. Hence,
by endowing the unitary group $\mathcal U\left(\bigoplus_{j=1}^m \hil\right)$ with Haar measure,
every $m$-tuple of C$^*$-convex coefficients is determined by a random unitary.
  
 Second, if $a=(a_1,\dots,a_m)$ is a tuple of C$^*$-convex coefficients,
 then $a$ induces a quantum channel $\mathcal E_a:\state(\hil)\rightarrow\state(\hil)$ via
 \[
 \mathcal E_a(\rho)\,=\,\sum_{j=1}^m a_j\rho a_j^*\,,\;\rho\in\state(\hil)\,.
 \]
 Conversely, every quantum channel induces a tuple of C$^*$-convex coefficients. Hence,
any source of randomness for channels (and there are many to choose from---see, for example,
\cite[\S14.7]{Bengtsson--Zyczkowski-book}), is a source of
randomness for nonclassical convex combinations.

We earlier introduced the notion of a proper C$^*$-convex hull $\cstarconvp\,R$ of a set $R$
by using C$^*$-convex combinations of elements of $R$ using only invertible C$^*$-convex coefficients.

\begin{definition} A \emph{coarsening} of quantum measurements $\nu_1,\dots,\nu_n\in\povm_\hil(X)$
is any measurement $\nu\in\povm_\hil(X)$ that satisfies
\[
\nu\,\in\,\cstarconvp\left(\{\nu_1,\dots,\nu_m\}\right)\,.
\]
\end{definition}

The use of invertible C$^*$-convex combinations in the definition of coarsening of measurements ensures 
that a certain level of information is conserved through the process of coarsening. 
 
We use the notation 
\[
\nu\,\ll_{\rm crse} \,(\nu_1,\dots,\nu_m)
\]
to indicate that $\nu$ is a coarsening of $\nu_1,\dots,\nu_m$.

\begin{definition} A quantum measurement $\nu$ is \emph{fine} if 
$\nu\,\ll_{\rm crse} \,(\nu_1,\dots,\nu_m)$ occurs only if each measurement $\nu_j$
is unitarily equivalent to $\nu$.
\end{definition}

Although the definition of fine measurement above appears to differ from the notion of sharp measurement defined earlier (Definition \ref{povm}),
the following theorem shows that the two concepts in fact are the same.

\begin{theorem}\label{cstar ext pt}
The following statements are equivalent for $\nu\in\povm_\hil(X)$:
\begin{enumerate}
\item\label{cstar ext-1} $\nu$ is a C$^*$-extreme point of $\povm_\hil(X)$;
\item\label{cstar ext-2} $\nu$ is fine;
\item\label{cstar ext-3} $\nu$ is sharp---that is,
there exist distinct $x_1,\dots,x_n\in X$ and pairwise-orthogonal projections $q_1,\dots,q_n\in\bh$
such that
\[
\nu\,=\,\sum_{j=1}^n \delta_{x_j}q_j\,.
\]
\end{enumerate}
\end{theorem}

The equivalence of \eqref{cstar ext-1} and \eqref{cstar ext-2} in Theorem \ref{cstar ext pt} is trivial, as the definition of fineness
herein is precisely
the definition of C$^*$-extreme point. Some preparatory results are required to show the third equivalence.

Via the identification of $C(X)$ as a unital
C$^*$-subalgebra of $C(X)\otimes\bh$, each $\phi\in\ucp_\hil(X)$ induces a ucp $\phi^{\rm c}:C(X)\rightarrow\bh$
by way of restriction:
\[
\phi^{\rm c}\,=\,\phi_{\vert C(X)}\,.
\]
In particular, if $\varrho_{x_0}\in\ucp_\hil(X)$ is spectral, then $\varrho_{x_0}^{\rm c}$ is in the character space of $C(X)$.

\begin{definition} {\rm (The $\Gamma^{\rm c}$-transform)} 
For each $\nu\in\povm_\hil(X)$ let $\Gamma^{\rm c}\nu$ denote the ucp map 
$C(X)\rightarrow\bh$ defined by
\[
\Gamma^{\rm c}\nu \,=\, \phi_\nu^{\rm c}\,.
\]
\end{definition}

\begin{proposition} $\Gamma^{\rm c}$ is properly C$^*$-affine. That is, 
\[
\Gamma^{\rm c}\left(\sum_{j=1}^mt_j^*\nu_jt_j\right)
\,=\,
\sum_{j=1}^mt_j^*\left(\Gamma^{\rm c}\nu_j\right)t_j
\]
for all $\nu_1,\dots,\nu_m\in\povm_\hil(X)$ all invertible C$^*$-convex coefficients $t_1,\dots,t_m\in\bh$.
\end{proposition} 

\begin{proof} Suppose first that $\nu\in\povm_\hil(X)$ and that $t\in\bh$ is invertible. We shall show
that
\begin{equation}\label{e:1}
\Gamma^{\rm c}(t^*\nu t)\,=\,t^*\left(\Gamma^{\rm c}\nu\right) t\,.
\end{equation}
To this end, let $\mu=\frac{1}{d}\tr\circ\nu$, $\tilde\nu=\frac{1}{d}\tr\circ(t^*\nu t)$, and $\tilde\mu=\frac{1}{d}\tr\circ\tilde\nu$. Clearly
$\tilde\mu \ll_{\rm ac} \mu$ because $\tilde\mu \ll_{\rm ac} \nu \ll_{\rm ac} \mu$. Conversely, if $\tilde \mu(E)=0$, 
then $t^*\nu(E)t=0$ and so $\nu(E)=0$ since $t$ is invertible. Thus, $\mu\ll_{\rm ac}\tilde\mu$.

For any $\psi\in C(X)$ and $\rho\in\state(\hil)$, we have
\[
\begin{array}{rcl}
\tr\left( \rho \displaystyle\int_X \psi\,d(t^*\nu t)\right)
&=& \displaystyle\int_X \psi\, \tr\left(\rho \displaystyle\frac{d(t^* \nu t)}{d\tilde \mu}\right) \,d\tilde\mu \\ && \\
&=& \displaystyle\int_X \psi\, \tr\left(\rho t^*\displaystyle\frac{d\nu}{d\tilde \mu}t\right) \,d\tilde\mu \\ && \\
&=& \displaystyle\int_X \psi\, \tr\left(\rho t^*\displaystyle\frac{d\nu}{d\mu}t\right) \displaystyle\frac{d\mu}{d\tilde\mu}\,d\tilde\mu \\ && \\
&=& \displaystyle\int_X \psi\, \tr\left(\rho t^*\displaystyle\frac{d\nu}{d\mu}t\right) \,d\mu \\ && \\
&=&\tr\left( \rho \left[t^*(\displaystyle\int_X \psi\,d\nu)t\right]\right)\,.
\end{array}
\]
Hence, \eqref{e:1} holds.

Suppose next that $\nu_1,\nu_2\in\povm_\hil(X)$ and that $\lambda_1,\lambda_2\in\mathbb R_+$ satisfy $\lambda_1+\lambda_2=1$. 
Let $\nu=\lambda_1\nu_1+\lambda_2\nu_2$, and let $\mu,\mu_1,\mu_2\in P(X)$ denote the induced probability measures. Hence, 
$\mu_j \ll_{\rm ac} \mu$ for $j=1,2$. Furthermore, 
for every $\psi\in C(X)$ and $\rho\in\state(\hil)$ we have that
\[
\begin{array}{rcl}
\displaystyle\int_X \psi\, \tr\left(\rho \displaystyle\frac{d\nu}{d\mu}\right)\,d\mu
&=&\displaystyle\int_X \psi\, \tr\left(\rho \left[ \lambda_1  \displaystyle\frac{d\mu_1}{d\mu} \displaystyle\frac{d\nu_1}{d\mu_1}\,+\,
\lambda_2\displaystyle\frac{d\mu_2}{d\mu}\displaystyle\frac{d\nu_2}{d\mu_2}\right]\right)\,d\mu 
\\ && \\
&=& \lambda_1 \displaystyle\int_X \psi\, 
\left[\tr\left(\rho  \displaystyle\frac{d\nu_1}{d\mu_1}\right)\right] \displaystyle\frac{d\mu_1}{d\mu}\,d\mu \,+\,
\lambda_2 \displaystyle\int_X \psi\, \left[\tr\left(\rho  \displaystyle\frac{d\nu_2}{d\mu_2}\right)\right] \displaystyle\frac{d\mu_2}{d\mu}\,d\mu
\\ && \\
&=&  \lambda_1 \displaystyle\int_X \psi\,\tr\left(\rho \displaystyle\frac{d\nu_1}{d\mu_1}\right)\,d\mu_1\,+\,
\lambda_2 \displaystyle\int_X \psi\,\tr\left(\rho \displaystyle\frac{d\nu_2}{d\mu_2}\right)\,d\mu_2\,.
\end{array}
\]
Thus, $\Gamma^{\rm c}(\lambda_1\nu_1+\lambda_2\nu_2)=\lambda_1\Gamma^{\rm c}\nu_1\,+\,\lambda_2\Gamma^{\rm c}\nu_2$. This fact
together with \eqref{e:1} implies that $\Gamma^{\rm c}$ is properly C$^*$-affine.
\end{proof}

A ucp map $C(X)\rightarrow\bh$ has, in principle, many different ucp extensions to $C(X)\otimes\bh$. However, if
a positive linear map $\varphi:C(X)\rightarrow\bh$ has the form 
\begin{equation}\label{nice form}
\varphi(g)=\sum_{j=1}^n g(x_j)b_j^*b_j
\end{equation}
for some distinct
$x_1,\dots,x_n\in X$ and C$^*$-convex coefficients $b_1,\dots, b_n\in\bh$, then the \emph{natural extension} 
$\mathfrak e(\varphi)$ of $\varphi$ from 
$C(X)$ to $C(X)\otimes \bh$ is defined by
\[
\mathfrak e(\varphi)\,=\,\sum_{j=1}^n b_j^{*}\varrho_{x_j}b_j\,.
\]
Observe that of $\varphi_1,\dots,\varphi_m:C(X)\rightarrow\bh$ are ucp maps of the form \eqref{nice form}, then
\begin{equation}\label{e is cstar affine}
\mathfrak e\left(\sum_{j=1}^mt_j^*\varphi_jt_j\right)
\,=\,
\sum_{j=1}^mt_j^*\mathfrak e(\varphi_j)t_j
\end{equation}
for all invertible C$^*$-convex coefficients $t_1,\dots,t_m\in\bh$.

A further direct consequence of the definition: if
$\nu\in\povm_\hil(X)$ has finite support, then 
\begin{equation}\label{fixed pts}
\Omega\circ\mathfrak e\circ\Gamma^{\rm c}\,\nu\,=\,\nu\,.
\end{equation}

We now turn to the proof of Theorem \ref{cstar ext pt}.

\begin{proof} Assume $\nu$ is a C$^*$-extreme point of $\povm_\hil(X)$. 
By Proposition \ref{cstar ext is ext}, $\nu$ is necessarily an extremal POVM. Thus,
by Theorem \ref{ext pt}, there are distinct points $x_1,\dots,x_n\in X$ and operators $h_1,\dots,h_n\in\bh_+$  such that
\[
\nu\,=\,\sum_{j=1}^n \delta_{x_j}h_j ,\quad
\phi_\nu\,=\,\sum_{j=1}^{n}h_j^{1/2}\varrho_{x_j}h_j^{1/2},
\quad\mbox{and}\quad 
\phi_\nu^{\rm c}\,=\,\sum_{j=1}^n \varrho_{x_j}^ch_j\,.
\]
Let $\hil_\pi=\bigoplus_1^n\hil$ and define $\pi:C(X)\rightarrow\bof(\hil_\pi)$ by $\pi=\bigoplus_{j=1}^n\varrho_{x_j}^{\rm c}$.
Thus, $\varphi_\nu^{\rm c}=w^*\pi w$, where $w:\hil\rightarrow\hil_\pi$ is the isometry $w\xi=\bigoplus h_j^{1/2}\xi$.
Let $\mathcal M$ be the commutant of $\pi\left(C(X)\right)$. Because $x_1,\dots, x_n$ are distinct, $\hil_\pi$ is generated by vectors of the
form $\pi(\psi)w\xi$, for $\psi\in C(X)$ and $\xi\in\hil$, and the commutant $\mathcal M$ is given by
$\mathcal M=\bigoplus_1^n \bh$. 

Suppose that $\phi_\nu^{\rm c}=\sum_{i=1}^m t_i^*\varphi_i t_i$ is a proper C$^*$-convex combination of ucp maps
$\varphi_i:C(X)\rightarrow\bh$. By the Radon-Nikod\'ym theorem for completely positive maps \cite[Theorem 1.4.2]{arveson1969},
for each $i$ there is a positive contraction $a_i\in \mathcal M$ such that $t_i^*\varphi_it_i=w^*a_i\pi w$. Because $a_i=\bigoplus_{j=1}^m a_j^i$
for  some $a_j^i\in\bh_+$, we obtain
\[
\varphi_i\,=\,\sum_{j=1}^n \varrho_{x_j}^{\rm c}(b_j^i)^*b_j^i,\quad\mbox{where } b_j^i=(a_j^i)^{1/2}h_j^{1/2}t_i^{-1}\,.
\]
Pass to the natural extension and apply $\Omega$; that is, let
\[
\nu_i\,=\,\Omega\left(\mathfrak e(\varphi_i)\right)\,=\,\sum_{j=1}^n \delta_{x_j}(b_j^i)^*b_j^i\,.
\]
Observe that $\Gamma^{\rm c}\nu_i=\varphi_i$ and
\[
\begin{array}{rcl}
\nu &=&\Omega\circ\mathfrak e(\phi_\nu^{\rm c})\,=\, \Omega\circ\mathfrak e \left( \displaystyle\sum_{i=1}^m t_i^*\Gamma^{\rm c}\nu_i t_i^*\right )
\\ && \\
&=& \Omega\circ\mathfrak e\circ\Gamma^{\rm c} \left( \displaystyle\sum_{i=1}^m t_i^* \nu_i t_i^*\right ) \\ && \\
&=& \displaystyle\sum_{i=1}^m t_i^*\nu_i t_i^*\,.
\end{array}
\]
Because $\nu$ is a C$^*$-extreme point, we obtain $\nu_i=u_i^*\nu u_i$ for some unitaries $u_1$,\dots,$u_m$ in $\bh$. Now apply $\Gamma^{\rm c}$
to obtain $\varphi_i=u_i^*\phi_\nu u_i$ for each $i$. This proves that $\phi_\nu$ is a C$^*$-extreme point in the space of all ucp maps $C(X)\rightarrow\bh$.
Because $C(X)$ is abelian and $\hil$ has finite dimension, all such extreme points are homomorphisms \cite[Proposition 2.2]{farenick--morenz1997},
\cite[Corollary 2.2]{farenick--zhou1998}. It is readily verified that $\phi_\nu^{\rm c}$ is multiplicative if and only if each $h_j$ is a projection and
$h_{j'}h_j=h_jh_{j'}=0$ for $j'\neq j$.

Conversely, suppose that $\nu$ is a projection-valued measure of the form $\nu\,=\,\displaystyle\sum_{j=1}^n \delta_{x_j}q_j$ for some distinct
 $x_1,\dots,x_n\in X$ and pairwise-orthogonal projections $q_1,\dots,q_n\in\bh$.
Suppose that $\nu=\sum_{i=1}^mt_i^*\nu_it_i$ is a proper C$^*$-convex combination. Because $t_i$ is invertible, $\nu_i \ll_{\rm ac} \nu$
(as we showed in the proof of Lemma \ref{res-support}), which implies that $\nu_i$ has finite support. Since $\Gamma^{\rm c}$ is properly
affine, $\phi_\nu^{\rm c}=\sum_{i=1}^m t_i^*\phi_{\nu_i}^{\rm c}t_i$. But because $\nu$ is projection-valued, the ucp map 
$\phi_\nu^{\rm c}:C(X)\rightarrow\bh$ is a homomorphism and, hence, C$^*$-extremal amongst all such ucp maps
\cite[Proposition 1.2]{farenick--morenz1997}. Thus, there are unitaries $u_1,\dots,u_m\in \bh$ such that $\phi_{\nu_i}^{\rm c}=u_i^*\phi_\nu u_i$
for all $i$. Hence, 
\[
\nu_i\,=\,\Omega\circ\mathfrak e\circ\Gamma^{\rm c}\nu_i\,=\,\Omega\left( \sum_{j=1}^n u_i^*q_j\varrho_{x_j}q_ju_i\right) 
\,=\,
u_i^*\nu u_i\,.
\]
That is, $\nu$ is a C$^*$-extreme point of $\povm_\hil(X)$.
\end{proof}

\subsection{Application: sharp measurements generate all quantum measurements through coarsening}

\begin{theorem}\label{km-thm} The C$^*$-convex hull of the C$^*$-extreme points of $\povm_\hil(X)$ is dense in $\povm_\hil(X)$.
That is, every quantum measurement is approximated by coarsenings of sharp measurements.
\end{theorem}

\begin{proof} The C$^*$-convex hull of the 
C$^*$-extreme points of $\povm_\hil(X)$ is,
by Theorem \ref{cstar ext pt}, the set of all C$^*$-convex combinations of sharp measurements.

Select $\nu\in\povm_\hil(X)$. By the Kre\v{\i}n-Milman Theorem, there is a net $\{\nu_\alpha\}_\alpha\subset\povm_\hil(X)$
such that each $\nu_\alpha$ is a convex combination of extreme points of $\povm_\hil(X)$ and $\nu_\alpha\rightarrow\nu$. Therefore, it
is sufficient to show that every extreme point of $\povm_\hil(X)$ is a coarsening of sharp measurements. Recall that Lemma \ref{fin-support}
asserts that an extreme point $\nu'$ of $\povm_\hil(X)$ must have finite support; that, is
$\nu'$ has the form $\nu'=\displaystyle\sum_{j=1}^m\delta_{x_j}h_j$
for some $x_1,\dots,x_m\in X$ and $h_1,\dots,h_m\in\eff(\hil)$.  By way of the (sharp) scalar-valued quantum probability
measures $E\mapsto \delta_{x_j}(E)1$ and the C$^*$-convex coefficients $a_j=h_j^{1/2}$ we obtain $\nu'=\displaystyle\sum_{j=1}^ma_j^*\delta_{x_j}a_j$,
which is a C$^*$-convex combination of C$^*$-extreme points of $\povm_\hil(X)$.
\end{proof}

%%%%%%%%%%%%%
\section{Discussion}

Our use of the terms ``classical'' and ``nonclassical'' in this paper inherently refer to ``scalar valued'' and "operator valued.'' As noted by
the referee, there is a highly nonclassical feature to what we are calling classical convexity. For example, a mixed state generally admits many distinct
decompositions as a (classical) convex combination of pure states, a fact which is studied in great detail in \cite{cassinelli--etal1997} and lies
at the heart of many of the difficulties in the interpretation of quantum mechanics \cite{Busch--Lahti--Mittelstaedt-book}.

Likewise, we use the term ``quantum measurement'' interchangeably with positive operator-valued probability measure. In this regard we
are following a common (as in \cite{Bengtsson--Zyczkowski-book}, for example) although not universal practice. A more refined terminology would use the term observable where
we have have used measurement, the term instrument for the next level in which the accompanying state changes are taken into account, and reserve the term
measurement for the highest level in which the entire description of the measurement model is given.

The integral representations afforded by the transforms $\Gamma$ and $\Omega$ are related to 
Fujimoto's cp-convexity \cite{fujimoto1994}. However, Fujimoto's cp-convexity is possibly too abstract to yield results as specific as those of the present paper. A more 
concrete yet still nonclassical notion of convexity is that of ``matrix convexity'' \cite{effros2009,webster--winkler1999}, which is slightly more general
than C$^*$-convexity. If for each $d\in\mathbb N$ one selects a $d$-dimensional Hilbert space $\hil_d$, then one defines 
\[
\povm(X)\,=\,\left(\povm_{\hil_d}(X)\right)_{d\in\mathbb N}\,,
\]
which is a matrix convex set. One may adapt the transforms $\Gamma$ and $\Omega$ to study $\povm(X)$ by way of 
unital completely positive linear maps $\psi:C(X)\otimes\bof(\hil_{d_1})\rightarrow\bof(\hil_{d_2})$
for arbitrary $d_1,d_2\in\mathbb N$. However, the Kre\v{\i}n-Milman Theorem in matrix
convexity \cite{webster--winkler1999} does not extend to C$^*$-convexity, making it necessary to
establish Theorem \ref{km-thm} herein by direct methods.

The transform $\Gamma^{\rm c}$ is well known. In the setting of Hilbert space, a good discussion is in Davies's book
\cite{Davies-book}---indeed, Theorem 4.1.2 of \cite{Davies-book} is especially relevant. A very general theory of $\Gamma^{\rm c}$ is achieved by
Ylinen's work on regular transformation measures \cite{ylinen2009}. Ylinen has not restricted his study to (finite-dimensional)
Hilbert space as we have done; he considers, more generally, arbitrary Banach and dual spaces in his framework. 

To endow $\povm_\hil(X)$ with a natural topology, we have opted to make use of the transforms $\Gamma$ and $\Omega$
rather than, as is done in \cite{Holevo-book2}, using the $\Gamma^{\rm c}$ transform. Our main reason for this preference is because, in our view,
$\Gamma^{\rm c}$ is a hybrid of classical and nonclassical notions, whereas $\Gamma$ is purely nonclassical.

The concept of coarseness is an order relation on $\povm_\hil(X)$ determined by quantum noise. In this regard,
a fine measurement is maximal with respect to the order. There are other orders of interest, such as those 
 related to cleanings \cite{clean2005} and smearings \cite{jencova--pulmannova2008} of measurements and
observables.

We have focused upon the case of compact $X$, but if $X$ is locally compact but not compact, then one may consider the abelian C$^*$-algebra
$C_0(X)$ of all continuous functions $X\rightarrow\mathbb C$ that vanish
at infinity. Let $\tilde X$ shall denote the one-point compactification of $X$. Thus, $C_0(\tilde X)=C(\tilde X)$, the unital C$^*$-algebra
of all continuous functions $\tilde X\rightarrow\mathbb C$. The C$^*$-algebra $C(\tilde X)\otimes \bh$ is isomorphic to the unitisation
$\left( C_0(X)\otimes\bh\right)^{\sim}$ of the non-unital algebra $C_0(X)\otimes\bh$.
If a linear map $\phi_0:C_0(X)\otimes\bh\rightarrow\bh$ is contractive and completely positive, then there is a unital completely positive (ucp) linear
map $\phi:C_0(\tilde X)\otimes\bh\rightarrow\bh$ extending $\phi_0$. Conversely, every ucp map $\phi:C_0(\tilde X)\otimes\bh\rightarrow\bh$
restricts to a contractive completely positive linear map $\phi_0:C_0(X)\otimes\bh\rightarrow\bh$. The set
\[
\left\{ \phi_{\vert C_0(X)\otimes\bh}\,\vert\, \phi \;\mbox{ is a ucp map }
C_0(\tilde X)\otimes\bh\rightarrow\bh\right\}
\]
plays the role of $\ucp_\hil\left(C(Y)\otimes\bh\right)$ for compact Hausdorff $Y$.
Because the passage from $X$ to $\tilde X$ amounts to nothing more than adjoining a unit to a nonunital C$^*$-algebra, one can make slight
reformulations of the results of the paper to cover the case of non-compact $X$.

Finally, because optimisation of the outcome statistics of apparatuses often amounts to minimising a real-valued concave function defined on the 
space $\eff(\hil)$ of quntum effects,
it seems appropriate to mention here that there is a nonclassical analogue, using the integral under study in the present paper, of the classical inequality of Jensen for convex functions.
For every selfadjoint $a\in\bh_{\rm sa}$ with spectrum in an open inteval $J\subset\mathbb R$, 
one may define a normal operator $\vartheta(a)\in\bh$, for a function $\vartheta:J\rightarrow\mathbb C$,
by functional calculus. Coupled with the L\"owner ordering of selfadjoint operators, one has the notion of operator convex function.
If $J\subset\mathbb R$ is an open interval, then a function 
$\vartheta:J\rightarrow\mathbb R$ is \emph{operator convex} if 
\[
\vartheta\left( \alpha a + (1-\alpha)b\right) \,\leq\,
\alpha\vartheta(a)+(1-\alpha)\vartheta(b)\,,
\]
for all $\alpha\in[0,1]$, all selfadjoint operators $a,b\in\bh_{\rm sa}$ with spectrum in $J$, and
all finite-dimensional Hilbert spaces $\hil$.

\begin{theorem}\label{jensen} {\rm (Jensen's Inequality in POVMs \cite{farenick--zhou2007})}
If $J\subset\mathbb R$ is an open interval containing a closed interval
$[\alpha,\beta]$, and if
$\kappa:X\rightarrow \bof(\hil)$ is a Borel-measurable function for which 
$\kappa(x)$ is hermitian and has spectrum contained in
$[\alpha,\beta]$
for every $x\in X$, then for any $\nu\in\povm_\hil(X)$ we have
\[
\vartheta\left(\int_X \kappa\,d\nu\right)\,\leq\,\int_X\vartheta\circ \kappa\,d\nu\,,
\]
for every operator convex function $\vartheta:J\rightarrow\mathbb R$.
\end{theorem}

In the case where $X$ is a finite sample space, Therorem \ref{jensen}
is the Hansen--Pedersen--Jensen Inequality \cite{hansen--pedersen2003}: for
any C$^*$-convex combination $\sum_{j=1}^m a_j^*y_ja_j$ of selfadjoint operators $y_1,\dots,y_m\in\bh$ with spectrum in an open interval $J$,
and for any operator convex function $\vartheta:J\rightarrow\mathbb R$, the following operator inequality holds:
\[
\vartheta\left( \sum_{j=1}^m a_j^*y_ja_j\right) \,\leq\, a_j^*\vartheta(y_j)a_j\,.
\]

%%%%%%%%%%%%%
\section{Conclusion}

In this paper we have studied the structure of the set $\povm_\hil(X)$ of quantum measurements of a quantum system (represented by a $d$-dimensional
Hilbert space $\hil)$ whose possible measurement events is the $\sigma$-algebra $\bor(X)$ of Borel sets of a compact Hausdorff space $X$.
The classical case occurs with $d=1$ and reduces to the study of probability measures. In classical analysis, one may integrate scalar-valued
Borel functions with respect to arbitrary probability measures; so doing produces a positive linear functional on the abelian C$^*$-algebra $C(X)$. 
Herein we have defined an integral so that one may integrate any quantum random variable against an arbitrary positive operator-valued measure, and
this has been achieved in a manner by which one produces a unital completely positive linear map of the homogeneous C$^*$-algebra
$C(X)\otimes\bh$ into the I${}_d$-factor $\bh$. Conversely, we have shown that there is a subclass of ucp maps
$\phi:C(X)\otimes\bh\rightarrow\bh$ such that each induces a positive operator-valued measure $\nu\in\povm_\hil(X)$. 

The transforms $\Gamma$ and $\Omega$ allow one to move between $\povm_\hil(X)$ and $\ucp_\hil(X)$. The transform $\Omega$ is C$^*$-affine,
which is sufficient structure to topologise $\povm_\hil(X)$ using the BW-topology of $\ucp_\hil(X)$ and to show that $\povm_\hil(X)$ is
a compact C$^*$-convex space. We have described precisely the structure of the extremal and C$^*$-extremal quantum measurements. The latter
are precisely the sharp observables, while the former are certain positive operator-valued measures with finite support and which were determined 
for finite and arbitrary $X$ by different methods in some earlier works \cite{chiribella--etal2007,chiribella--etal2010,dariano--etal2005,parthasarathy1999}.
As a consequence of the structure of extreme points and the Kre\v{\i}n-Milman Theorem, 
every quantum measurement that one can perform in principle 
can be approximated by quantum measurements
that one can perform in practice. That is, for every arbitrary quantum measurement $\nu$
(with perhaps infinitely many measurement outcomes) there is a quantum measurement $\nu'$ on a finite subsample space $X'\subset X$
in which the measurement statistics of the subsample approximate those 
of the general measurement $\nu$. 
By the nonclassical Kre\v{\i}n-Milman Theorem (Theorem \ref{km-thm}), 
the approximate $\nu'$ is a coarsening of a finite number of sharp measurements, 
each with measurement events $\bor(X')$.

%%%%%%%%%%%%%%%%%%
\section*{Acknowledgement}

We acknowledge the support of the NSERC Discovery, PGS, and USRA programs and 
Nipissing University (North Bay, Canada), where this work was undertaken during an extended
scientific visit of the first author. We are especially indebted to Giulio Chiribella for drawing our attention 
to the works \cite{chiribella--etal2007,chiribella--etal2010}, and to Michael Kozdron and the referee for useful commentary 
on the results herein.

%%%%%%%%%%%%%%%%%%%%%%%%%%% bibliography %%%%%%%%%%%%%%%%%%%%%%%%%

%%%%%%%%%%%%%%%%%%%%%%% END DOCUMENT %%%%%%%%%%%%%%
\end{document}